\documentclass{tlp}

\usepackage{amsfonts, amsmath, amssymb, mathrsfs}
\usepackage{times, helvet, courier}
\usepackage{graphicx, epsfig, epstopdf}
\usepackage{wrapfig}
\usepackage{multirow}
\usepackage{url}
\usepackage[linesnumbered,lined,commentsnumbered]{algorithm2e}
\usepackage{stmaryrd}

\usepackage{url}
\usepackage{todonotes}

\usepackage{times}
\usepackage{helvet}
\usepackage{courier}

\usepackage{latexsym}
\usepackage{amssymb}
\usepackage{graphicx}
\usepackage{acronym}

\def\naf{{ not \:}}

\newcommand{\argmax}{\operatornamewithlimits{argmax}}

\newboolean{includeMemo}
\setboolean{includeMemo}{true} 

\newcommand{\memo}[1]{
  \ifthenelse {\boolean{includeMemo}}{\medskip\noindent\fbox{\begin{minipage}[b]{\dimexpr\linewidth-1em}#1\end{minipage}}\medskip\newline}
}

\allowdisplaybreaks

\newtheorem{definition}{Definition} 
\newtheorem{example}{Example} 
\newtheorem{lemma}{Lemma}
\newtheorem{proposition}{Proposition}
\newtheorem{property}{Property}

\begin{document}
\bibliographystyle{acmtrans}

\long\def\comment#1{}

\title{Solving Distributed Constraint Optimization Problems Using Logic Programming}

\author[Tiep Le, Tran Cao Son, Enrico Pontelli, and William Yeoh]
{Tiep Le, Tran Cao Son, Enrico Pontelli, William Yeoh \\
Computer Science Department\\
New Mexico State University\\
Las Cruces, NM, 88001, USA\\ 
E-mail: \{tile, tson, epontell, wyeoh\}@cs.nmsu.edu
}

\pagerange{\pageref{firstpage}--\pageref{lastpage}}
\volume{\textbf{10} (3):}
\jdate{March 2002}
\setcounter{page}{1}
\pubyear{2002}

\maketitle

\label{firstpage}

\begin{abstract}
 This paper explores the use of \emph{Answer Set Programming (ASP)} in solving \emph{Distributed Constraint Optimization Problems (DCOPs).} The paper provides the following novel contributions: {\bf (1)}
 It shows how one can formulate DCOPs as logic programs; 
 {\bf (2)} It introduces ASP-DPOP, the first DCOP algorithm that is based on logic programming; 
 {\bf (3)} It experimentally shows that ASP-DPOP can be up to two orders of magnitude faster than DPOP (its imperative programming counterpart) as well as solve some problems that DPOP fails to solve, due to memory limitations; and 
 {\bf (4)} It demonstrates the applicability of ASP in a wide array of multi-agent problems currently modeled as DCOPs.\footnote{This article extends our previous conference paper~\cite{le:15} in the following manner: (1) It provides a more thorough description of the ASP-DPOP algorithm; (2) It elaborates on the algorithm's theoretical properties with complete proofs; and (3) It includes additional experimental results.} 
 
\noindent Under consideration in Theory and Practice of Logic Programming (TPLP).

\end{abstract}

\section{Introduction}

\emph{Distributed Constraint Optimization Problems (DCOPs)} are optimization
problems where agents need to coordinate the assignment of values to their ``local'' variables
to maximize the overall sum of resulting constraint utilities~\cite{modi:05,petcu:05,mailler:04,yeoh:12}.
The process is subject to limitations on the communication capabilities of the agents; in 
particular,  each agent can only
exchange information with neighboring agents within a given topology.
DCOPs are well-suited for modeling multi-agent coordination and resource allocation problems, where the primary interactions are between local subsets of agents. Researchers have used DCOPs to model various problems, such as the distributed scheduling of meetings~\cite{maheswaran:04a,zivan:14},  distributed allocation of targets to sensors in a network~\cite{farinelli:08},  distributed allocation of resources in disaster evacuation scenarios~\cite{lass:08b}, the distributed management of power distribution networks~\cite{kumar:09,jain:12}, the distributed generation of coalition structures~\cite{ueda:10} and the distributed coordination of logistics operations~\cite{leaute:11}.

The field has matured considerably over the past decade, since the seminal ADOPT paper~\cite{modi:05}, as researchers continue to develop more sophisticated  solving algorithms. The majority of the DCOP resolution
algorithms can be classified in one of three classes: {\bf (1)}~\emph{Search-based algorithms,} 
like ADOPT~\cite{modi:05} and its variants~\cite{yeoh:09a,yeoh:10,gutierrez:11,gutierrez:13}, AFB~\cite{gershman:09}, and MGM~\cite{maheswaran:04b}, where the agents enumerate  combinations of value assignments in a decentralized manner; 
{\bf (2)}~\emph{Inference-based algorithms,} like DPOP~\cite{petcu:05} and its variants~\cite{petcu:05b,petcu:07,petcu:07b,petcu:08}, max-sum~\cite{farinelli:08}, and Action GDL~\cite{vinyals:11b}, where the agents use dynamic programming techniques to propagate aggregated information to other agents; and 
{\bf (3)}~\emph{Sampling-based algorithms,} like DUCT~\cite{ottens:12} and D-Gibbs~\cite{nguyen:13,fioretto:14}, where the agents sample the search space in a decentralized manner.  

The existing algorithms have been designed and developed almost exclusively using  \emph{imperative programming} techniques, where the algorithms define a control flow, that is, a sequence of commands to be executed. In addition, the local solver employed by each agent is an ``ad-hoc'' implementation. In this paper, we are interested in investigating the benefits of using \emph{declarative programming} techniques to solve DCOPs, along with the use of a general constraint solver, used as a black box, as each agent's local constraint solver. Specifically, we propose an integration of \emph{Distributed Pseudo-tree Optimization Procedure (DPOP)}~\cite{petcu:05}, a popular DCOP algorithm, with \emph{Answer Set Programming (ASP)}~\cite{Niemela99,MarekT99} as the local constraint solver of each agent. 

This paper provides the first step in bridging the  areas of DCOPs and ASP; in the process, we  
offer novel contributions to both the DCOP field as well as the  ASP field. 
\emph{For the DCOP community,} we demonstrate that the use of ASP as a local constraint solver provides a number of benefits, including the ability to capitalize on {\bf (i)}~the highly expressive ASP language to more concisely define input instances (e.g.,~by representing constraint utilities as implicit functions instead of explicitly enumerating their extensions) and {\bf (ii)}~the highly optimized ASP solvers to exploit problem structure (e.g.,~propagating hard constraints to ensure consistency). 
\emph{For the ASP community,} the paper makes the equally important contribution of increasing the applicability of ASP to model and solve a wide array of multi-agent coordination and resource allocation problems, currently modeled as DCOPs. Furthermore, it also demonstrates that general, off-the-shelf ASP solvers, which are continuously honed and improved, can be coupled with distributed message passing protocols to outperform specialized imperative solvers. 

The paper is organized as follows. In Section~\ref{sec:background}, we review the basic definitions of  DCOPs, the DPOP algorithm, and ASP. 
In Section~\ref{sec:aspdpop},
we describe in detail the structure of the novel ASP-based DCOP solver, called \emph{ASP-DPOP,} and its implementation. Section~\ref{sec:analysis} provides
an analysis of the properties of ASP-DPOP,  including  proofs of soundness and  completeness of ASP-DPOP. 
Section~\ref{sec:expr} provides some experimental results, while  
Section~\ref{sec:relatedwork} reviews related work.  Finally,  Section~\ref{sec:conclusion} provides 
conclusions and indications for future work.
\section{Background}
\label{sec:background}
In this section, we present an overview of DCOPs, we describe DPOP, a complete distributed algorithm to solve DCOPs, and provide some fundamental definitions of ASP.

\subsection{Distributed Constraint Optimization Problems}

\noindent A \emph{Distributed Constraint Optimization Problem (DCOP)}~\cite{modi:05,petcu:05,mailler:04,yeoh:12} can be described as a tuple $\mathcal{M} = \langle \mathcal{X, D,F,A,\alpha} \rangle$ where:
\begin{itemize}
 \item $\mathcal{X} = \{x_1,\ldots,x_n\}$ is a finite set of (decision) \emph{variables}; 
 \item $\mathcal{D}=\{D_1,\ldots,D_n\}$ is a set of finite \emph{domains}, where  $D_i$ is the domain of 
 the variable $x_i \in \mathcal{X}$, for $1 \leq i\leq n$; 
 \item $\mathcal{F}=\{f_1,\ldots,f_m\}$ is a finite set of \emph{constraints}, where $f_j$
 is a  $k_j$-ary function 
 ${f_j:D_{j_1}\times D_{j_2} \times \ldots \times D_{j_{k_j}} \mapsto \mathbb{R} \cup \{-\infty\} }$ that specifies the utility of each combination of values of variables in its \emph{scope}; the scope is denoted by $scp(f_j)=\{x_{j_1},\ldots,x_{j_{k_j}}\}$;\footnote{For the sake of simplicity, we assume a given ordering of variables.} 
 \item $\mathcal{A}=\{a_1,\ldots,a_p\}$ is a finite set of \emph{agents}; and 
 \item $\alpha:\mathcal{X}\mapsto \mathcal{A}$ maps each variable to an agent.
 \end{itemize}
 
We say that a variable $x$ is owned by an agent $a$ if $\alpha(x) = a$. We denote with $\alpha_i$ the set of all variables that are owned by an agent $a_i$, i.e., 
$\alpha_i = \{x \in \mathcal{X}\:|\: \alpha(x) = a_i\}$.
 Each constraint in $\mathcal{F}$ can be either \emph{hard}, indicating that some value combinations result in a utility of $-\infty$ and must be avoided, or \emph{soft}, indicating that all value combinations result in a finite utility and need not be avoided. A \emph{value assignment} is a (partial or complete) function $x$ that maps variables of
 $\mathcal{X}$ to values in $\mathcal{D}$ such that, if $x(x_i)$ is defined, then $x(x_i) \in D_i$ for $i=1,\dots,n$. For the
 sake of simplicity, and with a slight abuse of notation,
 we will often denote $x(x_i)$ simply with $x_i$.
 Given a constraint $f_j$ and a \emph{complete value assignment} $x$ for all decision variables, we  denote with $x_{f_j}$ the projection of $x$ to the variables in $scp(f_j)$; we refer to this as a \emph{partial value assignment} for   $f_j$. For a DCOP $\mathcal{M}$, we denote
 with $\mathcal{C}(\mathcal{M})$ the set of all complete value assignments for $\mathcal{M}$.

A \emph{solution} of a DCOP is a complete value assignment ${\bf x}$ for all variables such that 
\begin{eqnarray}\label{DCOPmaximizeproblem}
{\bf x} =\argmax_{x \in \mathcal{C}(\mathcal{M})}\sum_{j=1}^{m} f_j(x_{f_j})
\end{eqnarray}

A DCOP can be described by its \emph{constraint graph}---i.e.,  a graph whose nodes correspond to agents in $\mathcal{A}$ and whose edges connect pairs of agents who own variables in the scope of the same constraint.

\begin{definition}[Constraint Graph]
A \emph{constraint graph} of a DCOP $\mathcal{M} = \langle \mathcal{X, D,F,A,\alpha} \rangle$ is an undirected graph $G_{\mathcal{M}}=(V,E)$ where $V = \mathcal{A}$ and 
\begin{eqnarray}
E = \{\{a, a'\} & \mid & \{a, a'\} \subseteq \mathcal{A}, \exists f \in \mathcal{F} \textit{ and } \{x_i, x_{j}\} \subseteq \mathcal{X}, \textit{ such that } \nonumber \\ 
&& \{x_i, x_{j}\} \subseteq scp(f), \textit{ and } \alpha(x_i) = a, \alpha(x_{j}) = a'\}.
\end{eqnarray}
\end{definition}

Given the constraint graph $G_{\mathcal{M}}$ and given a node $a\in \mathcal{A}$, we denote with $N(a)$ the
\emph{neighbors} of $a$, i.e., 
\begin{eqnarray}
N(a) = \{a' \in \mathcal{A} \:|\: \{a,a'\}\in E\}.
\end{eqnarray}

 
 \begin{definition}[Pseudo-tree]
A \emph{pseudo-tree} of a DCOP is a subgraph of $G_{\mathcal{M}}$ that has the same nodes as 
$G_{\mathcal{M}}$ such that \emph{(i)}~the included edges (called \emph{tree edges})  form a rooted tree, and \emph{(ii)}~two nodes that are connected to each other in $G_{\mathcal{M}}$ appear in the same branch of the tree.
\end{definition} 
The  edges of $G_{\mathcal{M}}$  that are not included in a pseudo-tree 
 are called \emph{back edges}. Notice that tree edges connect a node with its parent and its children, while back edges connect a node with its pseudo-parents and  pseudo-children---i.e.,  nodes closer to the root are parents or pseudo-parents, while those closer to the leaves are children or pseudo-children. A pseudo-tree of a DCOP can be constructed using distributed DFS algorithms~\cite{hamadi:98} applied to the constraint graph of the DCOP.

In this paper, we say that two variables are constrained to each other if they are in the scope of the same constraint. 
Given a pseudo-tree, the separator of a node $a_i$ is, intuitively, the set of variables that \emph{(i)} are owned by the ancestors of $a_i$, and \emph{(ii)} are constrained with some variables that are either owned by $a_i$ or the descendants of $a_i$. 
Formally, in a pseudo-tree, the \emph{separator} of a node $a_i$, denoted by $sep_i$, is:
\begin{eqnarray}
sep_i \!\!&\!\!=\!\!& \!\{ x_{i'} \in \mathcal{X} \mid  \alpha(x_{i'}) = a_{i'} \textit{ where } a_{i'} \textit{ is an ancestor of } a_i; \textit{ and } \nonumber \\ 
&& \hspace{1.2cm} \exists x_{i''} \in \mathcal{X}, f \in \mathcal{F}, \textit{ such that } 
a_{i''}\textit{ is either } a_i \textit{ or a descendant of } a_i, \nonumber \\
&& \hspace{3.9cm} \alpha(x_{i''}) = a_{i''}, \textit{ and } \{x_{i'}, x_{i''}\} \subseteq scp(f)\}
\end{eqnarray} 
We denote with $P_i$, $PP_i$, $C_i$, and $PC_i$ the parent,
the set of pseudo-parents,
 the set of children, 
 and the set of pseudo-children of a node $a_i$, respectively. For simplicity, if $A$ is a set of agents in $\mathcal{A}$, we also denote with $\alpha_A$  the set of variables in $\mathcal{X}$ that are owned by agents in $A$.

\begin{figure}[htbp]
  \begin{minipage}[h]{0.29\textwidth}
  \centering
  \includegraphics[scale=0.23]{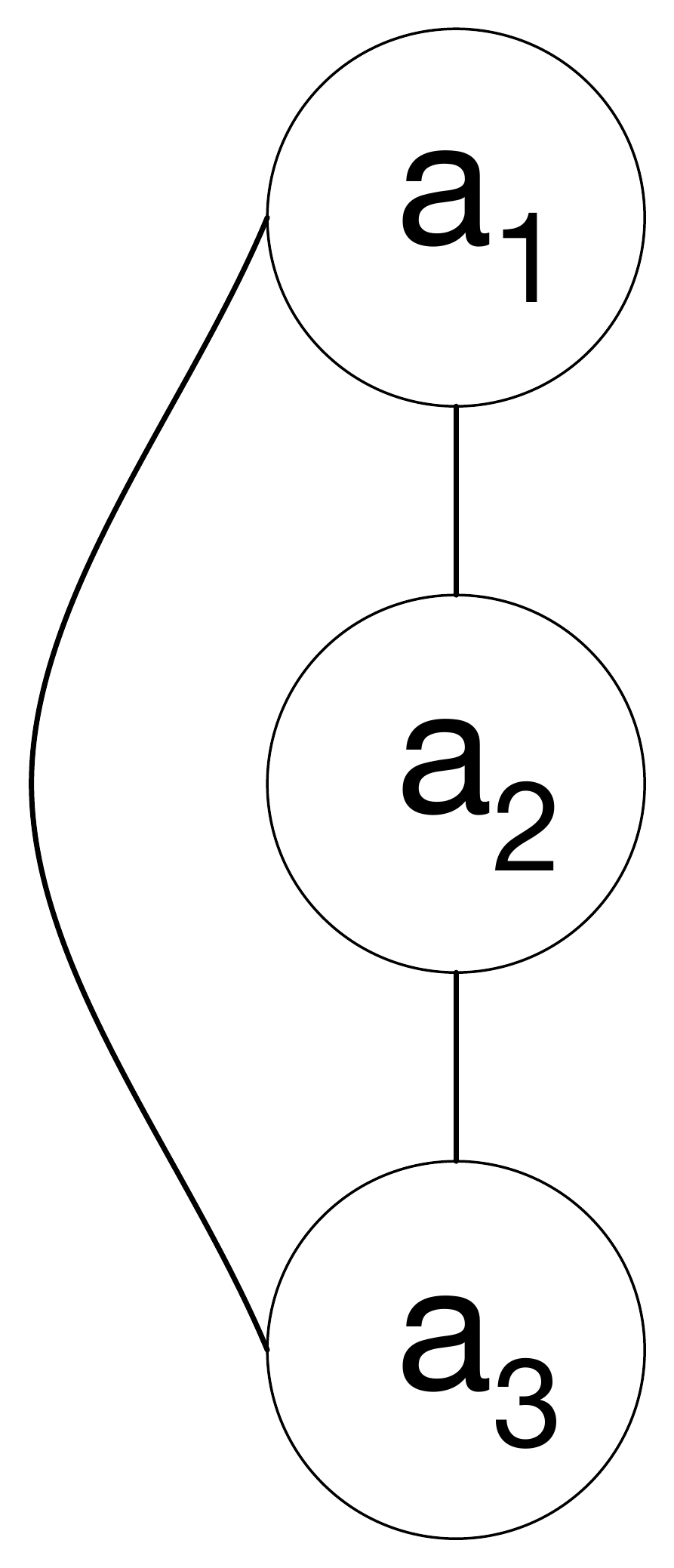}
  \\ \hspace{0.7cm}(a) Constraint Graph
  \end{minipage}
   \begin{minipage}[h]{0.29\textwidth}
  \centering  
  \includegraphics[scale=0.23]{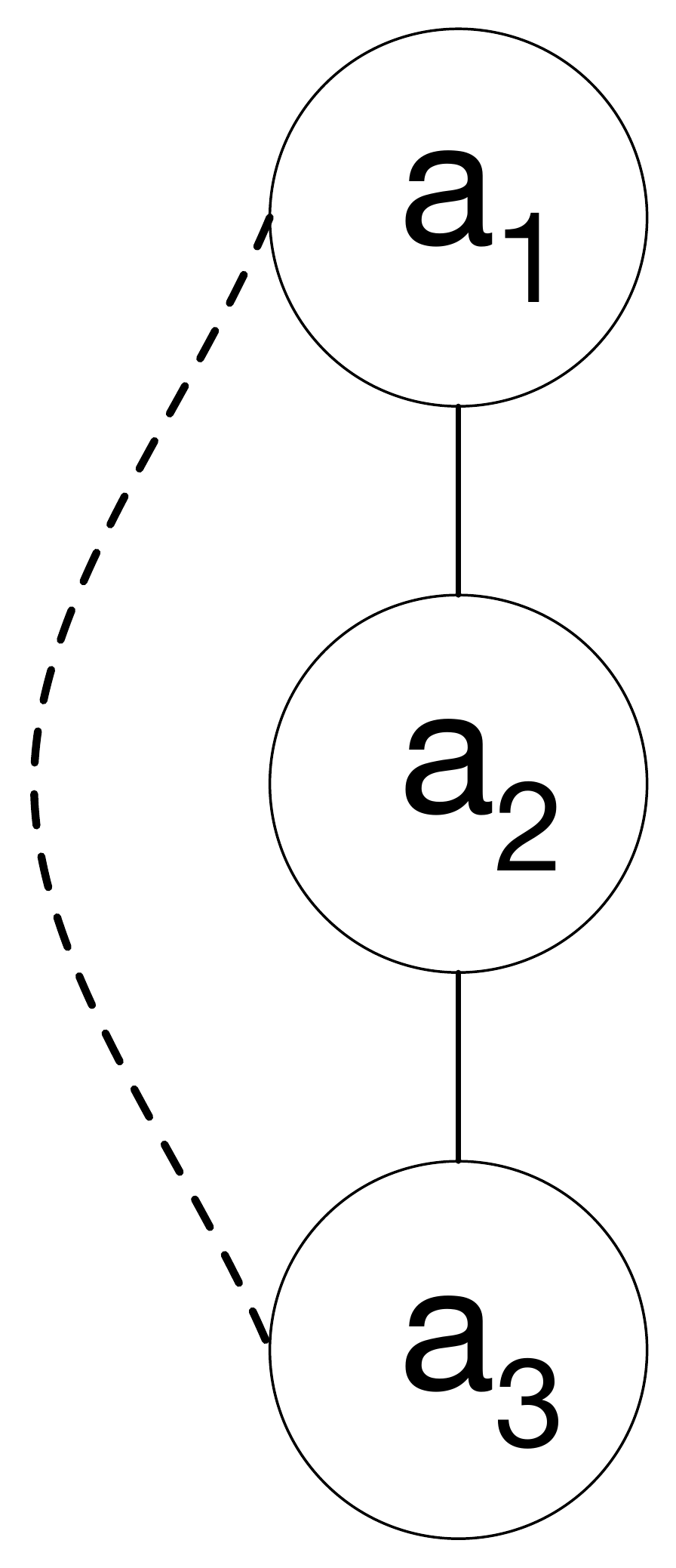}
  \\  \hspace{0.7cm} (b) Pseudo-tree
  \end{minipage}
  \hspace{-0.25in}
  \begin{minipage}[h]{0.45\textwidth}
    \centering
      \vskip0.5cm
    for $i < j$ \vspace*{-0mm} \\
    \begin{tabular}{|c|c|c|}
      \cline{1-3}
      $x_i$ & $x_j$ & Utilities \\
      \cline{1-3}
      0 & 0 & 5\\
      0 & 1 & 8\\
      1 & 0 & 20\\
      1 & 1 & 3\\
      \cline{1-3}
    \end{tabular}
    \\ \vskip2.6cm (c) Utilities of Constraints $xi\_cons\_xj$ \\
    with $i < j$
  \end{minipage}
  \vskip0.2cm
  \caption{Example DCOP \label{dcop}}
\end{figure}

\begin{example}
\label{exdcop}
Figure~\ref{dcop}(a) shows the constraint graph of a DCOP  $\mathcal{M} = \langle \mathcal{X, D,F,A,\alpha} \rangle$ where:
\begin{itemize}
 \item $\mathcal{X} = \{x_1, x_2, x_3\}$; 
 \item $\mathcal{D}=\{D_1, D_2, D_3\}$ where $D_i = \{0,1\}$ ($1\leq i \leq 3)$ is the domain of the variable $x_i \in \mathcal{X}$; 
 \item $\mathcal{F}=\{x1\_cons\_x2, \, x1\_cons\_x3, \, x2\_cons\_x3\}$ where, for each
	 $1 \leq i < j \leq 3$,
		\begin{itemize}
			\item for the constraint $xi\_cons\_xj$ we have that
				 $scp(xi\_cons\_xj) = \{x_i, x_j\}$;
			\item the utilities specified by the constraint $xi\_cons\_xj$ are given in Figure~\ref{dcop}(c).
		\end{itemize}
 \item $\mathcal{A}=\{a_1, a_2, a_3\}$; and 
 \item $\alpha$ maps each variable $x_i$ to  agent $a_i$.
 \end{itemize}
Figure~\ref{dcop}(b) shows one possible pseudo-tree, where  the dotted line is a back edge. In this pseudo-tree, $P_3 = a_2$, $PP_3 = \{a_1\}$, $C_1 = \{a_2\}$, $PC_1 = \{a_3\}$, and $sep_3 =\{ x_1, x_2 \}$.  
\end{example}

In a pseudo-tree $T$ of a DCOP $\langle \mathcal{X, D,F,A,\alpha} \rangle$, given $a_i \in \mathcal{A}$
let $R^T_{a_i}$  be the set of constraints in $\mathcal{F}$ such that:

\begin{eqnarray}\label{formula/R}
	R^T_{a_i} &=& \{f \in \mathcal{F} \mid scp(f) \subseteq \alpha_i \cup \alpha_{P_i} \cup \alpha_{PP_i} \wedge scp(f) \cap \alpha_i \not = \emptyset\}
\end{eqnarray} 

In the following, without causing any confusion, we often omit the superscript in $R^T_{a_i}$ (i.e., $R_{a_i}$) if there is only one pseudo-tree mentioned in the context.
\begin{example}
Considering again the DCOP in Example~\ref{exdcop} and its pseudo-tree in Figure~\ref{dcop}(b), we have $R_{a_3} = \{x1\_cons\_x3, \, x2\_cons\_x3\}$.
\end{example}

\subsection{The Distributed Pseudo-tree Optimization Procedure}

The \emph{Distributed Pseudo-tree Optimization Procedure (DPOP)}~\cite{petcu:05} is a complete algorithm to solve DCOPs with the following three phases:\footnote{Here we detail an extended version of DPOP described in~\cite{petcu:05} which removes the assumption that each agent owns exactly one variable.} Pseudo-tree
generation, UTIL propagation and VALUE propagation.
 
\subsubsection{Phase 1:  \emph{Pseudo-tree Generation Phase}} 

DPOP does not require the use of any specific algorithm to construct the pseudo-tree. However, in many implementations of DPOP, including those within the DCOPolis~\cite{sultanik:07} and FRODO~\cite{leaute:09} repositories, greedy approaches such as  the Distributed DFS algorithm~\cite{hamadi:98} are used.

The Distributed DFS algorithm operates as follows. First of all, the algorithm assigns a score to each agent, according to a heuristic function. It then selects an agent with the largest score as the root of the pseudo-tree. Once the root is selected, the algorithm initiates a DFS-traversal of the constraint graph, greedily adding the neighboring agent with the largest score as the child of the current agent. This process is
repeated until all agents in the constraint graph are added to the pseudo-tree. 

The agents' scores can be chosen arbitrarily. A commonly used heuristic is the max-degree heuristic $h(a_i)$:
\begin{align}
h(a_i) = | N(a_i) | 
\label{eq:max-degree}
\end{align}
which sets an agent's score to its number of neighbors. In situations where multiple agents have the same maximal score, the algorithm breaks ties according to a different heuristic, such as the variable-ID heuristic, which assigns to each agent a score that is equal to its unique ID. In our experiments, we use the max-degree heuristic and break ties with the variable-ID heuristic in the construction of the pseudo-tree. 

 
\subsubsection{Phase 2: \emph{UTIL Propagation Phase}}  

The UTIL propagation phase is a bottom-up process, which starts from the leaves of the
pseudo-tree and propagates upward, following only the tree edges of the pseudo-tree. In this process, the agents send UTIL messages to their parents. 


\begin{definition}
[UTIL Messages~\cite{DBLP:series/faia/2009-194}] 
\label{def3} 
$UTIL_{a_i}^{a_j}$, the UTIL message sent by  agent $a_i$ to  agent $a_j$, is a multi-dimensional matrix, with one dimension for each variable in $sep_i$. With a slight abuse of notation, we denote with $scp(UTIL_{a_i}^{a_j})$  the set of variables in the message. 
\end{definition}
Instead of using a multi-dimensional matrix, one can also flatten the multi-dimensional matrix into a table where each row of the table is for one combination of value assignment of variables in $sep_i$ and the respective utility for that combination. For simplicity, in this paper, we will represent UTIL messages under their tabular form. 
We can observe  that it is always true that $\alpha_j \cap scp(UTIL_{a_i}^{a_j}) \not = \emptyset$. The semantics of such a UTIL message is similar to a constraint whose scope is the set of all variables in the context of the message (its dimensions). The size of such a UTIL message is the product of the domain sizes of variables in the context of the message. 

Intuitively, a UTIL message summarizes the optimal sum of utilities in its subtree for each value combination of variables in its separator.
An agent $a_i$ computes its UTIL message by {\bf (i)}
 summing the utilities in the UTIL messages received from its child agents and the utilities of constraints whose scopes are exclusively composed of the  variables of $a_i$ and the variables in its separator (i.e., $R_{a_i}$), and then {\bf (ii)}
 projecting out  the  variables of $a_i$, by optimizing over them. Algorithm~\ref{algo1} provides a formal description of Phase 2. 

\RestyleAlgo{boxruled}
\LinesNumbered

\begin{algorithm}[htbp]
	{\small{
	Each agent $a_i$ does:\\ 
	$JOIN_{a_i}^{P_i} = null$\\
	\ForAll{$a_c \in C_i$}
	{ \label{line/3}
		wait for $UTIL_{a_c}^{a_i}$ message to arrive from $a_c$\\
		$JOIN_{a_i}^{P_i} = JOIN_{a_i}^{P_i} \oplus UTIL_{a_c}^{a_i}$ $//$ join UTIL messages from children as they arrive\\
	}
	$JOIN_{a_i}^{P_i} = JOIN_{a_i}^{P_i} \oplus \big (\oplus_{f \in R_{a_i}} f \big )$ $//$ also join all constraints with parent/pseudo-parents  \label{line/2}\\
	$UTIL^{P_i}_{a_i} = JOIN^{P_i}_{a_i} \bot_{\alpha_i}$ $//$ use projection to eliminate its owned variables \label{line/1}\\
	Send $UTIL_{a_i}^{P_i}$ message to its parent agent $P_i$\\
	\vskip0.3cm
	}}
	\caption{DPOP Phase 2 (UTIL Propagation Phase)}
\label{algo1}
\end{algorithm}

Algorithm~\ref{algo1} uses the \emph{JOIN} operator (i.e., $\oplus$) and the \emph{PROJECTION} operator (i.e., $\bot$).
\begin{definition}[JOIN $\oplus$ Operator] \label{joinoperation}
$U = UTIL_{a_k}^{a_i} \oplus UTIL_{a_l}^{a_i}$ is the join of two UTIL matrices (constraints). $U$ is also a matrix (constraint) with $scp(U) = scp(UTIL_{a_k}^{a_i}) \cup scp(UTIL_{a_l}^{a_i})$ as dimensions. For each possible combination $x$ of values of variables in $scp(U)$, the corresponding value of $U(x)$ is the sum of the corresponding cells in the two source matrices, i.e., $U(x) = UTIL_{a_k}^{a_i}(x_{UTIL_{a_k}^{a_i}}) + UTIL_{a_l}^{a_i}(x_{UTIL_{a_l}^{a_i}})$ where $x_{UTIL_{a_k}^{a_i}}$ and $x_{UTIL_{a_l}^{a_i}}$ are  partial value assignments from $x$ for all variables in $scp(UTIL_{a_k}^{a_i})$ and $scp(UTIL_{a_l}^{a_i})$, respectively.
\end{definition}
Since UTIL messages can be seen as constraints, the $\oplus$ operator can be used to join UTIL messages and constraints. 
\begin{example}\label{join32}
Given 2 constraints $x1\_cons\_x3$ and $x2\_cons\_x3$ in Example~\ref{exdcop}, let $JOIN_{a_3}^{a_2} = x1\_cons\_x3\, \oplus \, x2\_cons\_x3$. It is possible to see that ${scp(JOIN_{a_3}^{a_2})= \{x_1, x_2, x_3\}}$. The utility corresponding to $x_1 = x_2 = x_3 = 0$ is $JOIN_{a_3}^{a_2}(x_1 = 0, x_2 = 0, x_3 = 0) = 5 + 5  = 10$. Moreover, $JOIN_{a_3}^{a_2}(x_1 = 0, x_2 = 0, x_3 = 1) = 8 + 8  = 16$.
\end{example}

For the $\bot$ operator, knowing that $\alpha_i \subseteq scp(JOIN_{a_i}^{P_i})$, $JOIN_{a_i}^{P_i} \bot_{\alpha_i}$ is the projection through optimization of the $JOIN_{a_i}^{P_i}$ matrix along axes representing variables in $\alpha_i$. 

\begin{definition}[PROJECTION $\bot$ Operator] \label{projectionoperation}
Let $\alpha_i$ be a set of variables where $\alpha_i \subseteq scp(JOIN_{a_i}^{P_i})$, and let $X_i$ be the set of all possible value combinations of variables in $\alpha_i$. A matrix $U = JOIN_{a_i}^{P_i} \bot_{\alpha_i}$ is defined as: 
{\bf (i)} $scp(U) = scp(JOIN_{a_i}^{P_i}) \setminus \alpha_i$, and 
{\bf (ii)} for each possible value combination $x$ of variables in $scp(U)$, ${U(x) = \max_{x' \in X_i} JOIN_{a_i}^{P_i}(x,x')}$.  
\end{definition}

\begin{example}
Considering again $JOIN_{a_3}^{a_2}$ in Example~\ref{join32}, let $U = JOIN_{a_3}^{a_2} \, \bot_{\{x_3\}}$. We have $scp(U) = \{x_1, x_2\}$, and $U(x_1 = 0, x_2 =0) = \max \big (JOIN_{a_3}^{a_2}(x_1 = 0, x_2 = 0, x_3 =0), JOIN_{a_3}^{a_2}(x_1 = 0, x_2 = 0, x_3 =1) \big ) = \max ( 10, 16) = 16$.
\end{example}

As an example for the computations in Phase 2 (UTIL propagation phase), we consider again the DCOP in Example~\ref{exdcop}.  
\begin{example}
In the DCOP in Example~\ref{exdcop}, the agent $a_3$ computes its UTIL message, $UTIL_{a_3}^{a_2}$ (see Table~\ref{UTIL-DPOP}(a)), and sends it to its parent agent $a_2$. 
%
The agent $a_2$ then computes its UTIL message, $UTIL_{a_2}^{a_1}$ (see Table~\ref{UTIL-DPOP}(b)), and sends it to its parent agent $a_1$. Finally, the agent $a_1$ computes the optimal utility of the entire problem, which is $45$.
\end{example}

 \begin{table}[htbp]
 	\small
 	\centering
 	\begin{minipage}[h]{0.48\textwidth}
 		\centering
 		\resizebox{0.95\linewidth}{!} {
 			\begin{tabular}{|@{\hspace{0.8mm}}c@{\hspace{0.8mm}}| 
 					@{\hspace{0.2mm}}c @{\hspace{0.8mm}}| 
 					@{\hspace{0.2mm}}c 
 					@{\hspace{0.0mm}}r 
 					@{\hspace{0.0mm}}c 
 					@{\hspace{0.0mm}}r 
 					@{\hspace{0.0mm}}c 
 					@{\hspace{0.8mm}}r 
 					@{\hspace{0.0mm}}c 
 					@{\hspace{0.0mm}}r 
 					@{\hspace{0.0mm}}c 
 					@{\hspace{0.8mm}}|} 
 				\cline{1-11}
 				$x_1$ & $x_2$ & \multicolumn{9}{c|}{Utilities} \\
 				\cline{1-11}
 				0 & 0 & max(&5&+&5&, &8&+&8     &) = 16 \\
 				0 & 1 & max(&5&+&20&, &8&+&3   &) = 25 \\
 				1 & 0 & max(&20&+&5&, &3&+&8   &) = 25 \\
 				1 & 1 & max(&20&+&20&, &3&+&3 &) = 40 \\
 				\cline{1-11}
 			\end{tabular} 
 		}
 		\\ 
 		\vspace{0.25em} (a)
 	\end{minipage}
 	\begin{minipage}[h]{0.48\textwidth}
 		\centering
 		\resizebox{0.95\linewidth}{!} {
 			\begin{tabular}{|@{\hspace{0.8mm}}c@{\hspace{0.8mm}}| 
 					@{\hspace{0.8mm}}c 
 					@{\hspace{0.0mm}}r 
 					@{\hspace{0.0mm}}c 
 					@{\hspace{0.0mm}}r 
 					@{\hspace{0.0mm}}c 
 					@{\hspace{0.8mm}}r 
 					@{\hspace{0.0mm}}c 
 					@{\hspace{0.0mm}}r 
 					@{\hspace{0.0mm}}c 
 					@{\hspace{0.8mm}}|}

 				\cline{1-10}
 				$x_1$ & \multicolumn{9}{c|}{Utilities} \\
 				\cline{1-10}
 				0 & max( &5&+&16&, &8&+&25   &) = 33\\
 				1 & max( &20&+&25&, &3&+&40 &) = 45\\
 				\cline{1-10} 
 			\end{tabular} }\\ 
 			\vspace{2.25em} (b)
 		\end{minipage}
 		\caption{UTIL Phase Computations in DPOP\label{UTIL-DPOP}}
 	\end{table}

\subsubsection{Phase 3: \emph{VALUE Propagation Phase}}

Phase 2 finishes when the UTIL message reaches the root of the tree. At that point,
each agent, starting from the root of the pseudo-tree, determines the optimal value for its variables based on {\bf (i)} the computation from Phase 2, and 
{\bf (ii)} (for non-root agent only) the VALUE message that is received from its parent.  Then, it sends these optimal values to its child agents through VALUE messages. Algorithm~\ref{algo2} provides a formal description of Phase 3.

A VALUE message that travels from the parent $P_i$ to the
agent $a_i$, $\mathit{VALUE}_{P_i}^{a_i}$,  contains the optimal value assignment for variables owned by either the parent agent or the pseudo-parent agents of the agent $a_i$.

\begin{algorithm}[h]
	\small{
	Each agent $a_i$ do:\\ 
	wait for $\mathit{VALUE}_{P_i}^{a_i}(sep_i^*)$ message from its parent agent $P_i$ \hspace{.2cm} $//$ $sep_i^*$ is the optimal value assignment for all variables in $sep_i$  \\
		$\alpha_i^* \leftarrow \argmax_{\alpha_i \in X_i} JOIN_{a_i}^{P_i} (sep_i^*,\alpha_i)$ $//$ $X_i$ is the set of all possible value combinations of variables in $\alpha_i$  \label{computesolution}\\
	\ForAll{$a_c \in C_i$ }{
		let $sep_i^{**}$ be the partial optimal value assignment for variables in $sep_c$ from $sep_i^*$\\
		send $\mathit{VALUE}(sep^{**}_i, \alpha_i^*)$ as $\mathit{VALUE}_{a_i}^{a_c}$ message to its child agent $a_c$
	}
	\vskip0.3cm
	}
	\caption{DPOP Phase 3 (VALUE Propagation Phase)}
\label{algo2}
\end{algorithm}
\begin{example}
In the DCOP in Example~\ref{exdcop}, the agent $a_1$ determines that the value with the largest utility for its variable $x_1$ is 1, with a utility of 45, and then sends this information down to its child agent $a_2$ in a VALUE message, i.e., $\mathit{VALUE}_{a_1}^{a_2}(x_1 = 1)$. Upon receiving that VALUE message, the agent $a_2$ determines that 
the value for its variable $x_2$ with the largest utility of the subtree rooted at the agent $a_2$, assuming that $x_1 = 1$, is 0, with a utility of 45. The agent $a_2$ then sends this information down to its child agent $a_3$, i.e., $\mathit{VALUE}_{a_2}^{a_3}(x_1 = 1, x_2 = 0)$. Finally, upon receiving such VALUE message, the agent $a_3$ determines that the value for its variable $x_3$ with the largest utility of the subtree rooted at the agent $a_3$, assuming that $x_1 = 1$ and $x_2 = 0$, is 0, with a utility of 25.
\end{example}
\subsection{Answer Set Programming} 

Let us provide some general background on \emph{Answer Set Programming (ASP)} (see, for 
example, \cite{Baral03,michaelbook} for more details). 

 
An \emph{answer set program} $\Pi$ is a set of \emph{rules} of the form %
\begin{equation}
\label{lprule1}
c  \leftarrow a_1,\ldots,a_j,\naf a_{j+1},\ldots,\naf a_m
\end{equation}
where $0 \le j \le m$, for $1 \leq i \leq m$ each $a_i$  or $c$
is a literal of a first order language $\mathcal{L}$, and
$\naf$ represents {\em negation-as-failure (naf)}. 
For a literal  $a$,  
$\naf a$ is called a naf-literal. For a
rule of the form (\ref{lprule1}), the left and
right hand sides of the rule are called the \emph{head} and the
\emph{body} of the rule, respectively.  Both the head and the body can be empty. 
When the head is empty, the rule is called a {\em constraint}. 
When the body is empty, the rule is called a {\em fact}. A literal (resp.~rule) is a 
\emph{ground} literal (resp.~ground rule) if it does not contain any variable. A rule with variables is simply used as a shorthand for the set of its ground instances from the language $\mathcal{L}$. 
Similarly, a non-ground program (i.e., a program containing some non-ground rules) is a shorthand for all ground instances of its rules. 
%
Throughout this paper, we follow the traditional notation in writing ASP rules, where 
	names that start with an upper case letter represent variables.
%
For a ground instance $r$ of a rule of the form (\ref{lprule1}), 
$head(r)$ denotes the set $\{c\}$, while  
$pos(r)$ and $neg(r)$ denote $\{a_1,\ldots,a_j\}$ and 
$\{a_{j+1},\ldots, a_m\}$, respectively.  


Let $X$ be a set of ground literals. 
$X$ is consistent if there is no atom $a$ such that $\{a,\neg a\} \subseteq X$.
The
body of a ground rule $r$ of the form (\ref{lprule1}) is
{\em satisfied} by $X$ if $neg(r) \cap X = \emptyset$
and $pos(r) \subseteq X$. A ground rule of the form
(\ref{lprule1}) with nonempty head 
is satisfied by $X$ if either its body is not
satisfied by $X$ or $head(r) \cap X \ne \emptyset$. 
A constraint is {\em satisfied} by $X$ 
if its body is not satisfied by $X$. 

For a consistent set of ground literals $S$ and a ground 
program $\Pi$, the {\em reduct} 
of $\Pi$ w.r.t. $S$, denoted by $\Pi^S$, is the program
obtained from  $\Pi$ by deleting
({\bf i}) each rule that has a naf-literal {\em not} $a$ in its body where
$a \in S$, and
({\bf ii}) all naf-literals in the bodies of the remaining rules.

$S$ is an \emph{answer set (or a stable model)} of a ground program
  $\Pi$ \cite{GelfondL90} if it
satisfies the following conditions: 
{\bf (i)}  If $\Pi$ does not contain any naf-literal
(i.e., $j=m$ in every rule of $\Pi$) then $S$ is a minimal 
consistent set of literals that satisfies all the rules in $\Pi$; and
{\bf (ii)} If $\Pi$  contains some naf-literals
($j < m$ in some rules of $\Pi$) then $S$ is an answer set of $\Pi^S$. Note that $\Pi^S$ does
not contain naf-literals, and thus its answer set is defined in case {\bf (i)}.
A program $\Pi$ is said to be \emph{consistent} if it has some answer sets. Otherwise, it is inconsistent.

The ASP language includes also language-level extensions to facilitate the encoding of aggregates (\emph{min, max, sum, etc.}).
\begin{example}
Let us consider an ASP program $\Pi$ that consists of two facts and one rule:  
{
\begin{eqnarray}
int(5)&\leftarrow&\\
int(10)&\leftarrow&\\
max(U) & \leftarrow & U = \#max\{V:int(V)\} 
\end{eqnarray}
}
The third rule uses an aggregate to determine the maximum in the set $\{V\::\:int(V)\}$.
$\Pi$ has one answer set: $\{int(5), int(10), max(10)\}$. Thus, $\Pi$ is consistent.
\end{example}

Moreover, to increase the expressiveness of logic programming and  
simplify its use in applications, the syntax of ASP has been extended  with 
 \emph{choice rules}. Choice rules are of the form: 
{
\begin{eqnarray}
l \,\{a_1,\dots, a_m\}\, u \leftarrow a_{m+1},\ldots,a_n,\naf a_{n+1},\ldots,\naf a_k 
\end{eqnarray}
}
where $l \,\{a_1,\dots, a_m\}\, u$ is called a choice atom, 
$l$ and $u$ are integers, $l \leq u$, $0 \leq m \leq n \leq k$, and each $a_i$ is a literal for $1 \leq i \leq k$. This 
rule allows us to derive any subset of $\{a_1,\ldots,a_m\}$ whose cardinality is between the lower bound $l$ and upper bound $u$ whenever the body is satisfied. $l$ or $u$ can be omitted. If $l$ is omitted, $l = 0$, and if $u$ is omitted, $u = +\infty$.  
Standard syntax for choice rules has been proposed and adopted in most state-of-the-art ASP solvers, such as 
\textsc{clasp}~\cite{GebserKNS07} and \textsc{dlv}~\cite{dlv97}. 

\begin{figure}[htbp]
	\includegraphics[width=\textwidth]{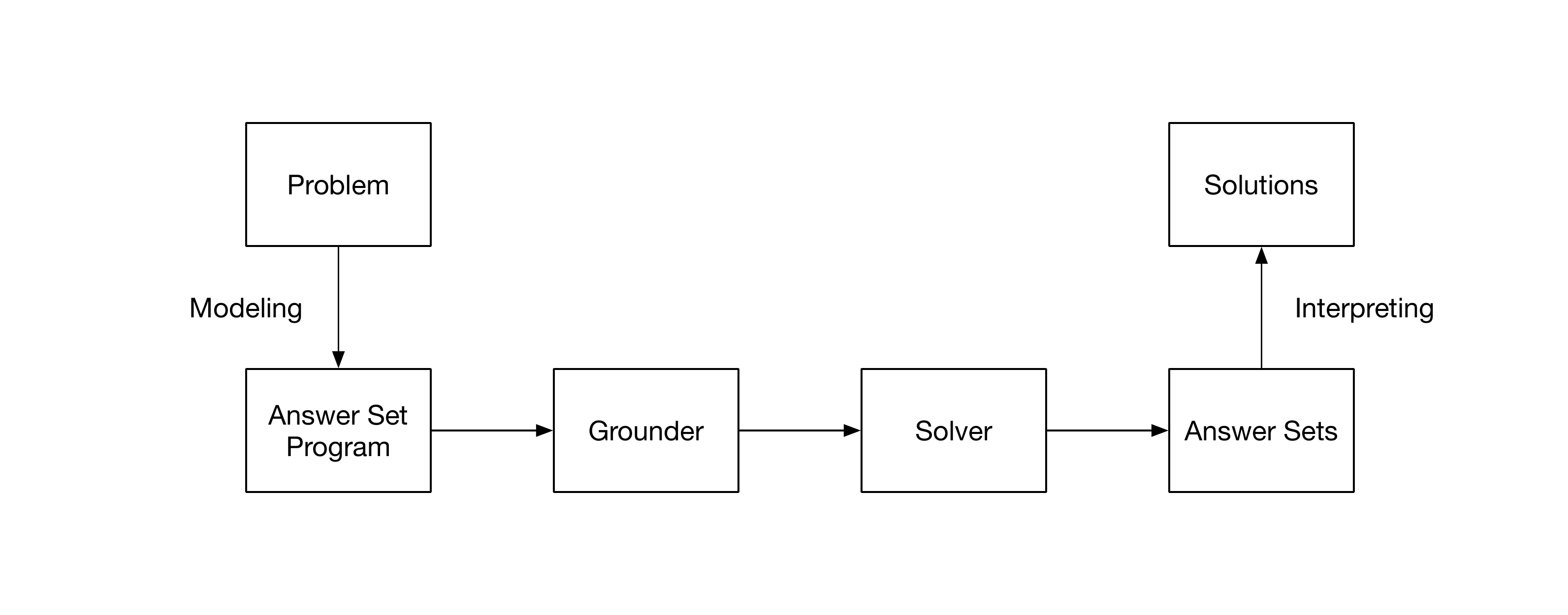}
	\caption{Solving a Problem Using ASP}
	\label{ASPSolver}
\end{figure}

Figure~\ref{ASPSolver} visualizes how to solve a problem using ASP. In more detail, the problem is encoded as an answer set program whose answer sets correspond to solutions. The answer set program, which may  contains variables, is then grounded using an ASP \emph{grounder}, e.g., \textsc{gringo}~\cite{Gebser:2011}. The grounding process employs smart techniques to 
reduce the size of the resulting ground program, e.g., removing literals from rules that are known to be true, removing 
rules that will not contribute to the computation of answer sets.

\begin{example}
\label{ex/gringo}
Let us consider an ASP program $\Pi$ that consists of two facts and one rule:
\begin{eqnarray}
int(1)&\leftarrow&\label{int1}\\
int(-1)&\leftarrow&\label{int-1}\\
isPositive(X)&\leftarrow&int(X), X > 0
\end{eqnarray}
Using a naive grounder that simply replaces consistently the variable $X$ with the two constants $1$ and $-1$, the ground program of $\Pi$  consists of the two facts \eqref{int1} and \eqref{int-1} and the two following ground rules:
\begin{eqnarray}
isPositive(1)&\leftarrow&int(1), 1 > 0\label{removeliterals1}\\
isPositive(-1)&\leftarrow&int(-1), -1 > 0 \label{unnecessary}
\end{eqnarray}
It is easy to see that the ground rule~\eqref{unnecessary} is unnecessary (i.e., its body cannot be satisfied by any set of literals due to the
literal $-1 > 0$) and should be removed. In contrast, the ground program of $\Pi$ obtained by \textsc{gringo} has only three facts:~\eqref{int1},~\eqref{int-1}, and 
\begin{eqnarray}
isPositive(1)&\leftarrow&\label{removeliterals2}
\end{eqnarray}
We observe that the unnecessary rule~\eqref{unnecessary} is omitted since its body cannot be satisfied (i.e., $-1 > 0$), and the fact~\eqref{removeliterals2} is obtained from the rule~\eqref{removeliterals1} by removing all literals in its body because the grounder can determine as been always satisfied. 
\end{example}
All the answer sets of the  program produced by the ASP grounder are then computed by an ASP \emph{solver}, e.g., \textsc{clasp}~\cite{GebserKNS07}. The solutions to the original problem can be determined by 
properly interpreting the different answer sets computed, where each answer sets corresponds to one
of the possible solutions to the original problem. 
 For readers who are interested in how to solve an answer set program, the foundations and algorithms underlying the grounding and solving technology used in \textsc{gringo} and \textsc{clasp} are
  described in detail in~\cite{claspbook,KaufmannLPS16}.  

\section{ASP-DPOP}\label{sec:aspdpop}
\emph{ASP-DPOP} is a framework that uses logic programming to capture the structure of DCOPs, and to emulate the computation and communication operations of DPOP. In particular, each agent in a DCOP 
is represented by a separate ASP program---effectively enabling the infusion of a knowledge
representation framework in the DCOP paradigm.
 
The overall communication infrastructure required by the distributed computation of DPOP
is expressed using a subset of the \emph{SICStus Prolog} language \cite{sicstus}, extended with multi-threading and the \emph{Linda} blackboard facilities. In ASP-DPOP, we use \textsc{clasp}~\cite{GebserKNS07}, with its companion grounder \textsc{gringo}, as our ASP solver, being the current state-of-the-art for ASP.  In this section, we will describe the structure of ASP-DPOP and its implementation.

\begin{figure}[htbp]
	\includegraphics[width=\textwidth]{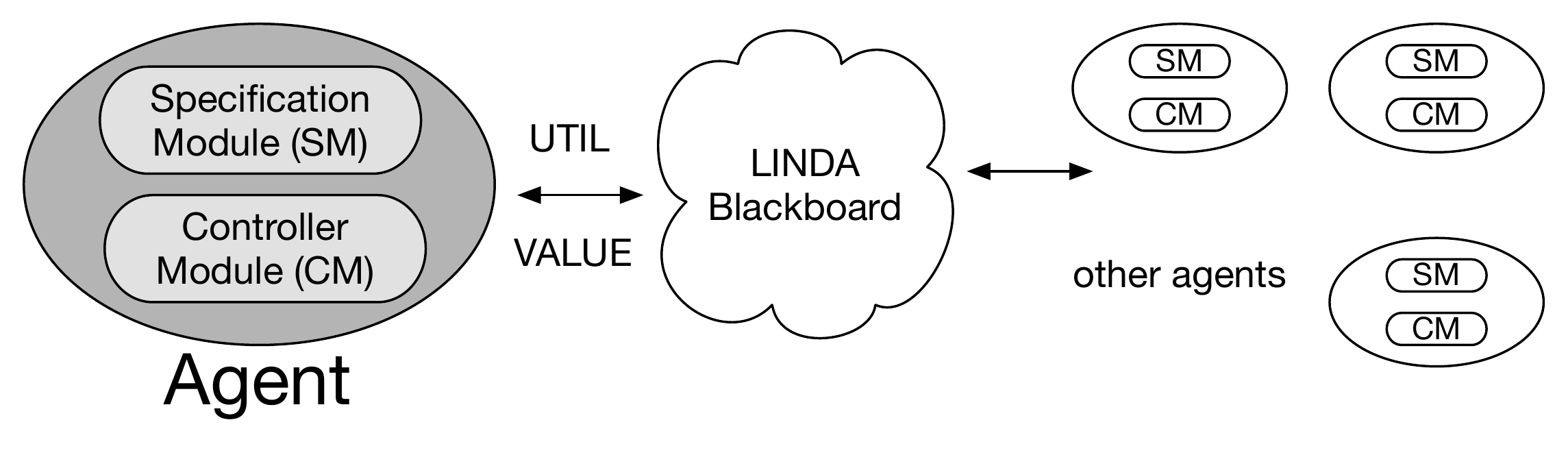}
	\caption{The structure of an ASP-DPOP agent}
	\label{ASPDPOPagentstructure}
\end{figure}

\subsection{The architecture of ASP-DPOP}

ASP-DPOP is an agent architecture that reflects the structure of DCOPs, where several  agents reflect the computation and communication operations of DPOP. The internal structure of each \emph{ASP-DPOP agent}, shown in Figure~\ref{ASPDPOPagentstructure}, is  composed of two modules. The first module is the \emph{Specification Module (SM),} that encloses an ASP program which captures a corresponding agent as  specified in the DCOP---i.e., the agent's name, the agent's neighbors, the 
description of the variables owned by the agent, the description of the variables owned by the agent's neighbors, and the description of the  constraints whose scope include any of the variables owned by the agent. 

The second module is the \emph{Controller Module (CM),} encoded as a Prolog program. The CM instructs the agent to 
perform the  communication operations of DPOP, such as cooperating with other agents to generate a pseudo-tree, waiting for UTIL messages from child agents,  sending the UTIL message to the parent agent (if present), waiting for the VALUE message from the parent agent (if present),  and sending the VALUE messages to the child agents. 

In ASP-DPOP, each DCOP is represented by a set of ASP-DPOP agents; each agent is modeled by its knowledge bases, located  at its SM and CM, and it interacts with other agents in accordance to the  rules of  its CM.

\subsection{ASP-DPOP Implementation: Specification Module (SM)}
Let us  describe how to capture the structure of a DCOP in the Specification Module of an ASP-DPOP agent using ASP. 
Let us consider a generic DCOP $\mathcal{M} = \langle \mathcal{X,D,F,A,\alpha} \rangle$. We represent $\mathcal{M}$ using a set of ASP-DPOP agents whose SMs are ASP programs $\{\Pi_{a_i} \mid a_i \in \mathcal{A}\}$. We will show how to generate $\Pi_{a_i}$ for each agent $a_i$. In the following, we say $a$ and $a'$ in $\mathcal{A}$ are \emph{neighbors} if there exists $x$ and $x'$ in $\mathcal{X}$ such that $\alpha(x) = a$, $\alpha(x') = a'$, and there is a $f \in \mathcal{F}$ such that $\{x, x'\} \subseteq scp(f)$. 
Given a constraint $f \in \mathcal{F}$, we say that $f$ is \emph{owned by the agent} $a_i$ if the scope of $f$ contains some variables owned by the agent $a_i$.\footnote{
The concept of  \emph{ownership of a constraint} is introduced to facilitate the representation of ASP-DPOP implementation. Intuitively, an agent should know about a constraint if the agent owns some variables that are in the scope of such constraint. Under this perspective, a constraint may be owned by several agents.
} 

For each variable $x_i \in \mathcal{X}$ we define a collection $L(x_i)$ of ASP rules that includes:
\begin{itemize}
\item A fact of the form
\begin{eqnarray}
	variable(x_i) \leftarrow
\end{eqnarray}
for identifying the name of the variable; 
\item For each $d \in D_i \in \mathcal{D}$, a fact of the form
\begin{eqnarray}
	value(x_i, d) \leftarrow \label{value}
\end{eqnarray}
for identifying the possible values of $x_i$. Alternatively, if the domain $D_i$ 
is an integer interval $[\textit{lower\_bound} \dots \textit{upper\_bound}]$ we can use the additional facts of the form
\begin{eqnarray}
	begin(x_i, lower\_bound) \leftarrow \label{domainlow} \label{begin}\\
	end(x_i, upper\_bound) \leftarrow \label{domaindown}\label{end}
\end{eqnarray}
to facilitate the description of the domain $D_i$.  In such case, the $value$ predicates similar to ones in~\eqref{value} are achieved by the rule
\begin{eqnarray}\label{BE}
	value(X, B..E) & \leftarrow & variable(X), begin(X, B), end(X, E)
\end{eqnarray}
Intuitively, $B$ and $E$ in~\eqref{BE} are variables that should be grounded with $lower\_bound$ and  $upper\_bound$ from~\eqref{begin} - \eqref{end}, respectively.
\end{itemize}
For each constraint $f_j \in \mathcal{F}$, where $scp(f_j) = \{x_{j_1},\ldots,x_{j_{k_j}}\}$, we define a collection $L(f_j)$ of rules that includes:
\begin{itemize}
\item A fact of the form
\begin{eqnarray}
	constraint(f_j) \leftarrow
\end{eqnarray} 
for identifying the name of the constraint; 
\item For each variable $x \in scp(f_j)$, a fact of the form
\begin{eqnarray}
\label{scope/eq}
	scope(f_j, x) \leftarrow
\end{eqnarray}
for identifying the scope of the constraint; and
\item For each 
partial value assignment $x_{f_j}$ for all variables in $scp(f_j)$,
 where $v_{j_1},\ldots,v_{j_{k_j}}$ are the value assignments of the variables $x_{j_1},\ldots,x_{j_{k_j}}$, respectively, 
  such that ${f_j(x_{f_j}) = u \not = -\infty}$, a fact of the form
\begin{eqnarray}
	f_j(u, v_{j_1},\ldots,v_{j_{k_j}}) \leftarrow \label{utility}
\end{eqnarray} 
For each 
partial value assignment $x_{f_j}$ for all variables in $scp(f_j)$,
 where $v_{j_1},\ldots,v_{j_{k_j}}$ are the value assignments of the variables $x_{j_1},\ldots,x_{j_{k_j}}$, respectively, 
  such that ${f_j(x_{f_j}) = -\infty}$, a fact of the form\footnote{$\#\mathit{inf}$ is a special constant representing the smallest possible value in ASP language.}
\begin{eqnarray}
	f_j(\#\mathit{inf}, v_{j_1},\ldots,v_{j_{k_j}}) \leftarrow \label{utilityinf}
\end{eqnarray}
\end{itemize}
Alternatively, it is also possible to envision cases where the utility of a constraint is
implicitly modeled by logic programming rules, as shown in the following example.
%
It is important to mention that, considering a constraint $f_j \in \mathcal{F}$:
\begin{itemize}
\item[{\bf(1)}] The order of variables (e.g., $x_{j_1},\ldots,x_{j_{k_j}}$) in $scp(f_j)$, whose corresponding value assignments (e.g., $v_{j_1},\ldots,v_{j_{k_j}}$) that appear in facts of the forms~\eqref{utility} and \eqref{utilityinf}, needs to be consistent in all facts of the forms~\eqref{utility} and \eqref{utilityinf} that relate to the constraint $f_j$; and
\item[{\bf(2)}] The order of the facts of the form~\eqref{scope/eq} that are added to $L(f_j)$ to identify the scope of the constraint $f_j$ needs to be consistent with the order of variables (e.g., $x_{j_1},\ldots,x_{j_{k_j}}$) mentioned in {\bf (1)}.
\end{itemize} 
These requirements (i.e., {\bf(1)} and {\bf(2)}) are critical, because they allow Controller Modules to understand which variables belong to what values that appear in the facts of the forms~\eqref{utility} and \eqref{utilityinf}, when Controller Modules read $L(f_j)$. This is done because, in \emph{SICStus Prolog}, the search rule is ``search forward from the beginning of the program.'' Therefore, the order of the predicates (i.e., facts) that are added to \emph{SICStus Prolog} matters.

\begin{example}
Let us consider a constraint $f$ whose scope is $\{x, x'\}$, and $f$ specifies that the utility of value assignments $x = v$, $x' = v'$ is $v+v'$.
%
  The facts of the form~\eqref{utility} for the constraint $f$ can be modeled by the following rule
\begin{eqnarray}
	f(V+V', V, V' ) \leftarrow value(x, V), value(x', V')
\end{eqnarray}
\end{example}

\smallskip
\noindent For each agent $a_i$ we define an ASP program $\Pi_{a_i}$ that includes:
\begin{itemize}
\item A fact of the form
\begin{eqnarray}
	agent(a_i) \leftarrow
\end{eqnarray}
for identifying the name of the agent;
\item For each variable $x \in \mathcal{X}$ that is owned by the agent $a_{i}$, a fact of the form
\begin{eqnarray}
	owner(a_i, x) \leftarrow
\end{eqnarray}
\item For each agent $a_{j}$ who is a neighbor of the agent $a_i$, a fact of the form
\begin{eqnarray}
	neighbor(a_{j}) \leftarrow
\end{eqnarray}
\item For each variable $x' \in \mathcal{X}$ that is owned by an agent $a_{j}$ who is a neighbor of the agent $a_i$, a fact of the form
\begin{eqnarray}
	owner(a_{j}, x') \leftarrow 
\end{eqnarray}
\item For each constraint $f_j \in \mathcal{F}$ owned by the agent $a_i$, the set of  rules 
\begin{eqnarray}
	L(f_j)
\end{eqnarray}
\item For each variable $x\in \mathcal{X}$ that is in the scope of some constraints owned by the agent $a_i$, the set of  rules 
\begin{eqnarray}
	L(x)
\end{eqnarray}
\end{itemize}
\subsection{ASP-DPOP Implementation: Encoding UTIL and VALUE Messages} \label{encode_UTIL_VALUE_mess}
The ASP-DPOP framework emulates the computation and communication operations of DPOP, where each ASP-DPOP agent produces UTIL and VALUE messages, and forwards them to its parent and child agents, as DPOP does. In ASP-DPOP, UTIL and VALUE messages are encoded as ASP facts, as discussed in this subsection. 

\subsubsection{UTIL Messages} 

In DPOP, each UTIL message sent from a child agent $a_i$ to its parent agent $P_i$ is a  matrix. In encoding a UTIL message
in ASP-DPOP, we represent each cell of the matrix of the UTIL message, whose associated utility is not $-\infty$, as an ASP 
atom of the form:
\begin{eqnarray}\label{utilfact}
	table\_max\_a_i(u, v_{i_1}, \ldots, v_{i_{k_i}}) 
\end{eqnarray} 
which indicates that the optimal aggregate utility of the value assignments $x_{i_1} = v_{i_1}, \ldots, x_{i_{k_i}} = v_{i_{k_i}}$ is $u \not = -\infty$, where  $sep_i = \{x_{i_1}, \ldots, x_{i_{k_i}}\}$. In other words, the parent agent $P_i$ knows that $UTIL_{a_i}^{P_i}(x_{i_1} = v_{i_1}$, \dots, $x_{i_{k_i}} = v_{i_{k_i}}) = u \not = -\infty$ if it receives the fact $table\_max\_a_i(u, v_{i_1}, \ldots, v_{i_{k_i}})$. It is important to know that the encoding of a UTIL message omits the cells whose associated utilities are $-\infty$.

In addition to facts of the form \eqref{utilfact}, $a_i$ also informs $P_i$ about variables in its separator. 
%
Thus, the encoding of the  UTIL message sent from the agent $a_i$ to the agent $P_i$ includes also
  atoms of the form:
\begin{eqnarray}
	& table\_\mathit{info}(a_i, a_{i_1}, x_{i_1}, lb_{i_1}, ub_{i_1}) \label{tableinfo1}\\
	& \cdot \cdot  \cdot \nonumber \\
	& table\_\mathit{info}(a_i, a_{i_{k_i}}, x_{i_{k_i}}, lb_{i_{k_i}} ub_{i_{k_i}}) \label{tableinfo2}
\end{eqnarray}
Each fact $table\_\mathit{info}(a_i, a_{i_t}, x_{i_t}, lb_{i_t}, ub_{i_t})$ informs $P_i$ that $x_{i_t}$ is a variable in the separator of $a_{i}$ whose domain is specified by  $lb_{i_t}$ (lower bound) and $ub_{i_t}$ (upper bound) and whose owner is 
$a_{i_t}$.\footnote{For simplicity, we assume that the domains $D_i$ are integer intervals.}
%
It is also critical to note that the order of the atoms of the forms~\eqref{tableinfo1} - \eqref{tableinfo2} matters, since such order will 
	allow the respective Controller Module understand which variable belongs to the values stated in facts of the form~\eqref{utilfact} after reading such encoded UTIL messages.
	
%

\begin{example}\label{utilmessexample}
Consider again the DCOP in Example~\ref{exdcop}. The UTIL message, sent from the agent $a_3$ to the agent $a_2$, in Table~\ref{UTIL-DPOP}(a) is encoded as the set of the ASP atoms:
\begin{eqnarray}
&&table\_max\_a_3(16, 0, 0) \\
&&table\_max\_a_3(25, 0, 1) \\
&&table\_max\_a_3(25, 1, 0) \\
&&table\_max\_a_3(40, 1, 1) \\
&&table\_\mathit{info}(a_3, a_1, x_1, 0, 1)\\
&&table\_\mathit{info}(a_3, a_2, x_2, 0, 1)
\end{eqnarray}
\end{example}
\begin{example}
Similarly, considering again the DCOP in Example~\ref{exdcop}, the UTIL message sent from the agent $a_2$ to the agent $a_1$ in Table~\ref{UTIL-DPOP}(b) is encoded as the set of ASP facts:
\begin{eqnarray}
&&table\_max\_a_2(33, 0) \\
&&table\_max\_a_2(45, 1) \\
&&table\_\mathit{info}(a_2, a_1, x_1, 0, 1)
\end{eqnarray}
\end{example}

\subsubsection{VALUE Messages} 

In DPOP, each VALUE message sent from a parent agent $P_i$ to its child agents $a_i$ contains the optimal value assignment for variables owned by either the parent agent or the pseudo-parent agents of the agent $a_i$. Thus, in encoding a VALUE message, we use atoms of the form:
\begin{eqnarray}
	solution(a, x, v)
\end{eqnarray}  
where $v$ is the value assignment of the variable $x$ owned by the agent $a$ for an optimal solution.
\begin{example}
Consider again the DCOP in Example~\ref{exdcop}. The VALUE message sent from the agent $a_1$ to the agent $a_2$ is encoded as the ASP atom:
\begin{eqnarray}
&&solution(a_1, x_1, 1)
\end{eqnarray}
Similarly, the VALUE message sent from the agent $a_2$ to the agent $a_3$ is encoded as the set of the ASP atoms:
\begin{eqnarray}
&&solution(a_1, x_1, 1)\\
&&solution(a_2, x_2, 0)
\end{eqnarray}
\end{example}
\subsection{ASP-DPOP Implementation: Controller Module (CM)}
\label{controller}
\begin{figure}[t]
\vspace{-1em}
\centering
\includegraphics[scale=0.62]{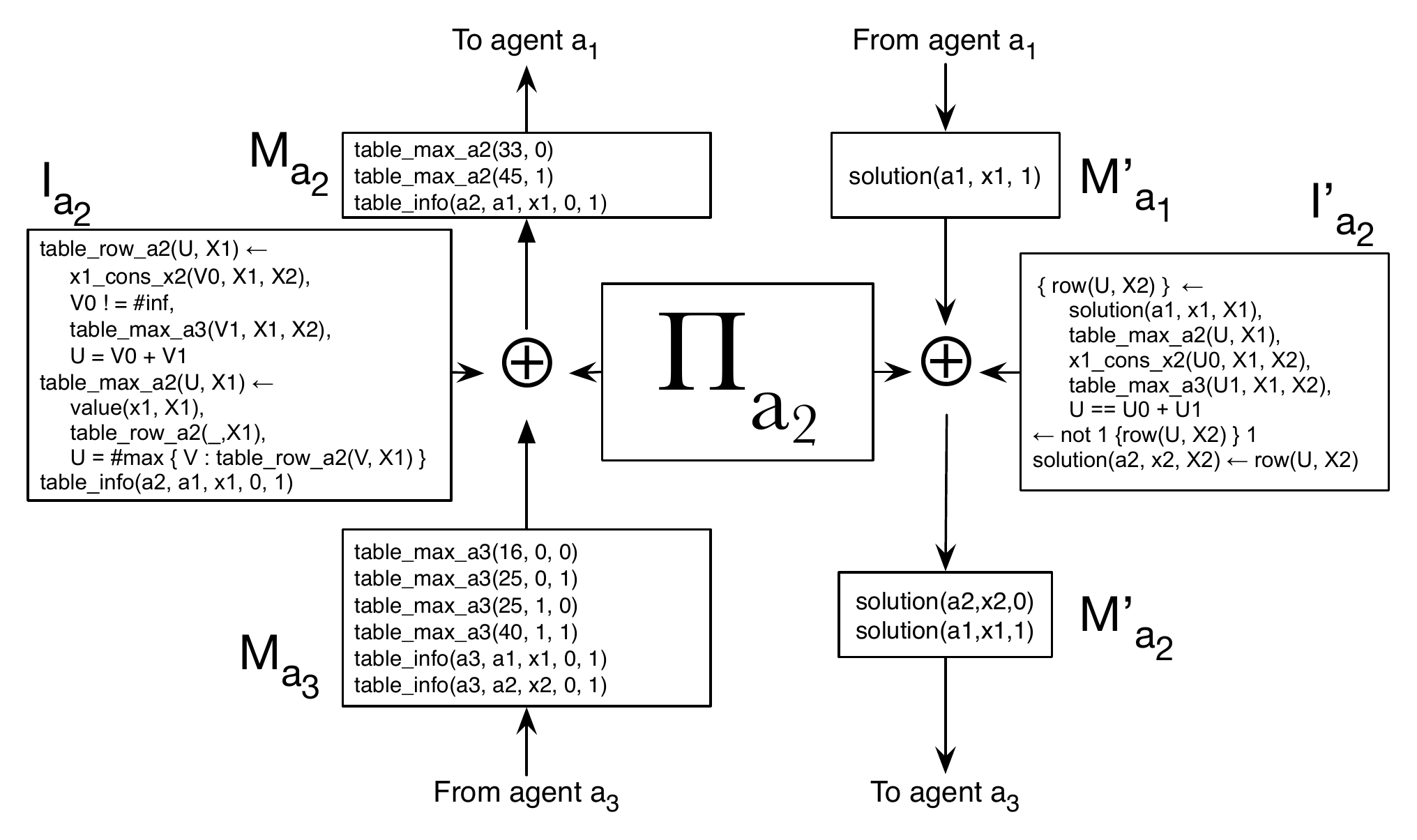}
\caption{Phase 2 and Phase 3 of Agent $a_2$ in ASP-DPOP on DCOP in Example~\ref{exdcop}}
\label{flow2}
\end{figure} 
The controller module in each ASP-DPOP agent $a_i$, denoted by $C_{a_i}$, consists of a set of Prolog
rules for communication (sending, receiving, and interpreting messages) and a set of rules for generating an ASP program that is used for the computations of a UTIL message and a VALUE message. In this subsection, we would like to discuss some code fragments to show how $C_{a_i}$ is structured.\footnote{The code listed in this section is a simplified version of the actual code for $C_{a_i}$, showing a condensed version of the clause bodies; however, it still gives a flavor of the implementation of $C_{a_i}$ and should be sufficiently explanatory for the purpose of the controller module.} To begin with, we will show how $C_{a_i}$ uses the Linda blackboard library of Prolog to exchange the messages. 

There are three types of messages exchanged through the Linda blackboard; they are \emph{tree}, \emph{util}, and \emph{value} messages that are used in Phase 1, Phase 2, and Phase 3,
%
respectively, of DPOP. For sending (resp.~waiting for) a message, we use the built-in
predicate $out/1$ (resp.~$in/1$) provided by the Linda library of Prolog. Every message is formatted as 
$message(From, To, Type, Data)$ where the arguments denote the agent who sends this message, the agent who should receive this message, the type of the message, and the data enclosed in the message, respectively. The implementation of the communication and the three phases of DPOP are described next.

\subsubsection{Sending Messages}  
The following Prolog rule generates a message of  type $t \in \{{tree, util, value}\}$, with  content $d$ ({\tt Content}), to be sent
 from an agent $a_i$ ({\tt From}) to an agent $a_{k}$ ({\tt To}):
 
{\small
\begin{verbatim}
% sending message
send_message(From,To,Type,Content) :- 
                   out(message(From,To,Type,Content)).
\end{verbatim}
} 

\subsubsection{Waiting for Messages}  
The following Prolog rule instructs agent $a_k$ ({\tt a\_k}) to wait for a message: 

{\small 
\begin{verbatim}
% waiting for a message
wait_message(From,a_k,Type,Data):- in(message(From,a_k,Type,Data)).
\end{verbatim}
} 
\noindent where $From, Type$, and $Data$ will be unified with the name of the agent who sent this message, the type of the message, and the data enclosed in the message, respectively.  

\subsubsection{Creating the Pseudo-Tree: Phase 1} \label{subsubphase1}




In this phase, ASP-DPOP agents cooperate with each other to construct a pseudo-tree.
%
For simplicity, we will show here the clauses in $C_{a_i}$ for generating a pseudo-tree by initiating a  DFS-traversal. We assume that the agent $a_i$ is not the root of the pseudo-tree. The agent $a_i$ waits for a $tree$ message from an agent $Parent$. The content  ($Data$) enclosed in such a $tree$ message is the set of \emph{visited} agents---i.e., the agents who have already started performing the DFS. Upon receiving a $tree$ message, $a_i$ will execute the following clauses:

{\small 
\begin{verbatim}
% pseudo-tree generation
generate_tree(Parent, Data):- 
             assign_parent(Parent),
             assign_pseudo_parent(Data),
             append(Data, [a_i], NewData), 
             depth_first_search(Parent, NewData).
       
% performing depth first search
depth_first_search(Parent, Data):-      
             find_next(Data, Next_Agent),
             (Next_Agent == none -> 
                     send_message(a_i, Parent, tree, Data) 
                     ;
                     assign_child(Next_Agent), 
                     send_message(a_i, Next_Agent, tree, Data),
                     wait_message(Next_Agent, a_i, tree, NewData),
                     depth_first_search(Parent, NewData)
             ).              
\end{verbatim}
}
\noindent Intuitively, upon receiving a $tree$ message from the agent $Parent$ enclosed with data $Data$, the agent $a_i$ will execute the clause $\texttt{generate\_tree(Parent, Data)}$. Specifically:
\begin{itemize}
	\item It executes the clause $\texttt{assign\_parent}/1$, where it adds to its $\Pi_{a_i}$ a fact of the form $parent(Parent)$;
	\item It executes the clause $\texttt{assign\_pseudo\_parent}/1$ which  adds to its $\Pi_{a_i}$ facts of the form $pseudo\_parent(X)$, where $X$ is a neighboring agent of  $a_i$ that appears in $Data$ such that $X \neq Parent$;
	\item It adds itself (i.e., a$\_$i) to the list of visited agents;
	\item It starts performing a DFS, by executing  the rule $\texttt{depth\_{first}\_search}/2$.
\end{itemize} 
In order to perform a DFS, the agent $a_i$ will execute the rule $\texttt{{find}\_next}/2$ to select
 a neighboring agent that will be  visited next; this selection is based on some heuristics 
 (i.e., the unvisited neighbor agent with the greatest number of neighbors). 
 If such an agent $Next\_Agent$ exists (i.e., $Next\_Agent \neq none$), then $a_i$ will:
 \begin{itemize}
 	\item Execute the rule $\texttt{assign\_child}/1$, used to add to its $\Pi_{a_i}$ a fact of the form $children(Next\_Agent)$;
 	\item Send a $tree$ message to the agent $Next\_Agent$;
 	\item Wait for the reply message from the agent $Next\_Agent$, which will provide  the updated list $NewData$ of visited agents;
 	\item  Recursively execute the rule $\texttt{depth\_{first}\_search}/2$.
 \end{itemize}
 Otherwise, if there is no agent $Next\_Agent$  (i.e., $Next\_Agent$ is equal to $none$), 
 then the agent $a_i$ will reply a $tree$ message to its agent $Parent$. This implies that the
  agent $a_i$ has  finished performing DFS at its branch. 

When the agent $a_i$ is chosen to be the root of the pseudo-tree, it executes the rule
 $\texttt{generate\_tree}(master, [\,])$ immediately without waiting for the $tree$ message from other agents. We note that an agent whose parent agent is $master$ will be the root of the pseudo-tree. It is also worth to notice that, at the end of this phase, the information about the parent, pseudo-parents, and child agents of each agent $a_i$ are added to $\Pi_{a_i}$ via facts of the forms $parent/1$, $pseudo\_parent/1$, and $children/1$, respectively. 
\begin{lemma}
\label{lemma/phase1linear}
In ASP-DPOP, Phase 1 requires a linear number of messages in $n$ where $n$ is the number of agents.
\end{lemma}
\begin{proof}
We first prove that Phase 1 terminates. In fact, each agent $a_i$ in executing $\texttt{depth\_{first}\_search/2}$ will perform the rule $\texttt{{find}\_next}/2$ to select
 a neighboring agent, i.e., $Next\_Agent$, that is {\em not} in the set of visited agents to send a {\em tree} message to. $Next\_Agent$ can be seen as an unvisited neighboring agent of the agent $a_i$. 
 
 The agent $a_i$ then waits to receive the {\em tree} message from $Next\_Agent$ enclosing an updated set of visited agents, and again send a {\em tree} message to another unvisited neighboring agent if it exists. 
  We notice that the updated set of visited agents will be expanded with at least one agent that is $Next\_Agent$ since $Next\_Agent$ will add itself to the set of visited agents beyond receiving the {\em tree} message from the agent $a_i$. 
  
  Therefore, every agent $a_i$ will send at most $|N(a_i)|$ {\em tree} messages to its child agents, where $N(a_i)$ is the set of the neighboring agents of the agent $a_i$. If there is no unvisited neighboring agent left, the agent $a_i$ will send a {\em tree} message to its parent agent together with the most updated set of visited agents, and terminates executing $\texttt{depth\_{first}\_search/2}$. Thus, it terminates performing  the clause $\texttt{generate\_tree(Parent, Data)}$. As a consequence, we can conclude that Phase 1 terminates. 
   
Furthermore, considering a pseudo-tree that is generated at the end of Phase 1. We can realize that the set of visited agents which are passing among agents is expanded with a non-root agent if and only if there is a {\em tree} message sent from a parent agent to its child agent downward the pseudo-tree. It is worth to remind that the agent who is selected to be the root of the pseudo-tree adds itself to the set of visited agents at the beginning. Thus, there are $n-1$ {\em tree} messages that are sent downward the pseudo-tree. Moreover, every agent except the root agent will send exactly one {\em tree} message to its parent agent upward the pseudo-tree. Therefore, there are $n-1$ {\em tree} messages that are sent upward the pseudo-tree. Accordingly, in total there are $2\times (n-1)$ {\em tree} messages produced in Phase 1. This proves Lemma~\ref{lemma/phase1linear}.
\end{proof}
\subsubsection{Computing the UTIL Message: Phase 2} \label{phase2ASPDPOP}
 
In the following, for simplicity, given an agent $a_i$, we assume that $a_p = P_i$. We will use $a_p$ and $P_i$ interchangeably. In this phase, each ASP-DPOP agent generates an ASP program for computing the UTIL message that will be sent to its parent.  In more detail, each agent  $a_i$ executes the following clause:
 
 
{\small 
\begin{verbatim}
% Phase 2: UTIL Propagation Phase
perform_Phase_2(ReceivedUTILMessages):-
         compute_separator(ReceivedUTILMessages, Separator),
         assert(separatorlist(Separator)),
         compute_related_constraints(ConstraintList), 
         assert(constraintlist(ConstraintList)),
         generate_UTIL_ASP(Separator, ConstraintList),
         solve_answer_set1(ReceivedUTILMessages, Answer),
         store(Answer),
         send_message(a_i, a_p, util, Answer).
\end{verbatim}
}
%
\noindent In particular, each agent $a_i$ with a set of  child agents $C_i$: 
\begin{itemize}
\item Waits to receive all of the UTIL messages from its children and combines them into a set of ASP facts.
Let $M_{a_k}$ be the  encoding of the UTIL message $UTIL^{a_i}_{a_k}$. We define a list $Received\mathit{UTIL}Messages$ as
follows. 
 
{
\begin{eqnarray}\label{getUMess}
	Received\mathit{UTIL}Messages = \bigcup_{a \in C_i}  M_{a}.
\end{eqnarray}
}
When $a_i$ is a leaf ($C_i = \emptyset$), we set $Received\mathit{UTIL}Messages = [\: ]$. 
 
\item Computes its separator $sep_i$ by executing $\texttt{compute\_separator}/2$. This is 
realized using  {\bf (i)} the information about its parent and pseudo-parent agents added in $\Pi_{a_i}$ during
 Phase 1, and {\bf (ii)} the information about ancestors of the agent $a_i$ that are directly connected with descendants of the agent $a_i$, via facts of the form $table\_\mathit{info}$, contained in the UTIL messages received from its child agents; 
\item Computes  the set $R_{a_i}$ ($ConstraintList$) of the related constraints (i.e., executing $\texttt{compute\_related\_constraints}/1$) that is defined as~\eqref{formula/R}. 
%
\item Generates the information for its UTIL message (i.e., executing   $\texttt{generate\_{UTIL}\_ASP}/2$). Specifically, 
$\texttt{generate\_{UTIL}\_ASP}/2$ first creates a logic program, denoted by $I_{a_i}$, from the separator list, the constraint list, and the information from $M_{a_k}$ where $a_k \in C_i$. It then computes the answer set of $ \Pi_{a_i} \cup I_{a_i} \cup  (\bigcup_{a \in C_i}   M_{a })$ which contains the encoded UTIL message of the agent $a_i$.  Assume that 
\begin{itemize}
\item $sep_i = \{x_{s_1}, \dots, x_{s_{k}}\}$ (i.e., $Separator = [x_{s_1}, \dots, x_{s_k}]$ is the separator list of $a_i$);
\item $R_{a_i} = \{f_{r_1}, \dots, f_{r_{k'}}\}$ and $scp(f_{r_j}) = \{x_{r_{j_1}},\dots, x_{r_{j_w}}\}$ for $1 \leq j \leq k'$ (i.e., $ConstraintList = [f_{r_1}, \dots, f_{r_{k'}}]$);
\item $C_{i} = \{a_{c_1}, \dots, a_{c_l}\}$ and each $UTIL_{a_{c_t}}^{a_i}$ has $x_{c_{t_1}}, \dots, x_{c_{t_{z}}}$ as its dimensions for $1 \leq t \leq l$; 
and
\item $a_p$ is the parent agent of the agent $a_i$.
%
\end{itemize}   
$\texttt{generate\_{UTIL}\_ASP}/2$ creates $I_{a_i}$ with the following rules: 

{\small \begin{equation}
\begin{array}{l}
\texttt{table\_row}\_{a_i}(U, X_{s_1}, \dots, X_{s_k})   \leftarrow  \label{summingrule}\\
	\hspace{2.5cm}  f_{r_1}(V_{r_1}, X_{r_{1_1}},\dots, X_{r_{1_w}}), \\
	\hspace{2.5cm} \cdots \\
		\hspace{2.5cm} f_{r_{k'}}(V_{r_{k'}}, X_{r_{k'_1}},\dots, X_{r_{k'_w}}),  \\
		\hspace{2.5cm} V_{r_{1}} \,!\!= \,\texttt{\#inf}, \cdots, V_{r_{k'}} \,!\!= \,\texttt{\#inf},  \\
		\hspace{2.5cm} \texttt{table\_max}\_a_{c_1}(U_{c_1}, X_{c_{1_1}},\dots, X_{c_{1_{z}}}),  \\
		\hspace{2.5cm}\cdots \\
		\hspace{2.5cm} \texttt{table\_max}\_a_{c_l}(U_{c_l}, X_{c_{l_1}},\dots, X_{c_{l_{z}}}),  \\
		\hspace{2.5cm} U = V_{r_1} + \dots + V_{r_{k'}} + U_{c_1} + \dots + U_{c_l}. 
\end{array}
\end{equation}}

{\small \begin{equation}
\begin{array}{ll}
\texttt{table\_max}\_a_i(U, X_{s_1}, \dots, X_{s_k})  \leftarrow   \label{maxrule}\\
	\hspace{2.5cm} \texttt{value}(x_{s_1}, X_{s_1}),\\
	\hspace{2.5cm} \cdots\\
	\hspace{2.5cm} \texttt{value}(x_{s_k}, X_{s_k}), \\	
	\hspace{2.5cm} \texttt{table\_row}\_a_i(\_, X_{s_1}, \dots, X_{s_k}),\\
	\hspace{2.5cm} U = \texttt{\#max}\{V : \texttt{table\_row}\_a_i(V, X_{s_1}, \dots, X_{s_k}) \}.  
\end{array}
\end{equation}}

{\small
\begin{equation}
\begin{array}{l}
\texttt{table}\_\texttt{info}(a_i, a_{s_1}, x_{s_1}, lb_{s_1}, ub_{s_1}). \hspace{4.5cm} \label{info1}\\
\cdots\\
\texttt{table}\_\texttt{info}(a_i, a_{s_k}, x_{s_k}, lb_{s_k}, ub_{s_k}).  
\end{array}
\end{equation}}

%
$\texttt{generate\_{UTIL}\_ASP}/2$ uses the information in $UTIL^{a_i}_{a_{c_t}}$ 
and $\Pi_{a_i}$ to generate the facts \eqref{info1}. In addition, each variable in the rules \eqref{summingrule}-\eqref{info1} corresponds to a variable name (e.g., $X_{s_1}$ corresponds to $x_{s_1}$ in the separator list; $X_{c_{1_1}}$ corresponds to $x_{c_{1_1}}$ in dimensions of $UTIL_{a_{c_1}}^{a_i}$; etc.). Therefore, due to the definition of the separator of  $a_i$ and $sep_i = \{x_{s_1}, \dots, x_{s_{k}}\}$, $X_{s_1}, \dots, X_{s_k}$ are guaranteed to occur on the right hand side of \eqref{summingrule}. In other words, $I_{a_i}$ is a {\em safe} program.

Intuitively, the rule of the form~\eqref{summingrule} creates the joint table for $a_i$---that is similar to the result of flattening $\big (\oplus_{a_{c_t} \in C_i} UTIL^{a_i}_{a_{c_t}} \big ) \oplus  \big (\oplus_{f \in R_{a_i}} f \big )$ into a table---
 given $UTIL^{a_i}_{a_{c_t}}$ and $R_{a_i}$. 
 In addition,~\eqref{maxrule} computes the optimal utilities for each value combination of variables in the separator list.

\item Computes an answer set $A$ of the program $\Pi_{a_i} \cup I_{a_i} \cup \bigcup_{a \in C_i} M_a $ by executing $\texttt{solve\_answer\_set1}/2$, and extracts from $A$ the information for the  UTIL message  (i.e., $Answer$) that will
be sent from the agent $a_i$ to the agent $a_p$. 

\item Asserts the information in $Answer$ for later use in Phase 3 (i.e., executes $\texttt{store}(Answer)$).   

\item Sends encoded $UTIL^{a_p}_{a_i}$ to the parent agent $a_p$ (i.e., executes $\texttt{send\_message}/4$). 
\end{itemize}
 

\begin{example}\label{ex/Ia2}
As an example, we refer to the DCOP in Example~\ref{exdcop}. 
Specifically, we illustrate $I_{a_2}$ generated by the agent $a_2$. 
$Received\mathit{UTIL}Messages$ for the agent $a_2$ is the set of ASP facts given in Example~\ref{utilmessexample}, $Separator = [x_1]$, 
and $ConstraintList = [x1\_cons\_x2]$. The program $I_{a_2}$ includes the following rules:
\begin{eqnarray*}
\texttt{table\_row}\_a_2(U, X_1) & \leftarrow  &x_1\_cons\_x_2(V_0, X_1, X_2),  \\
					&&			    V_0 \,\, !\!= \#\mathit{inf}, \nonumber\\
					&&	                     \texttt{table\_max}\_a_3(V_1, X_1, X_2), \nonumber\\
					&&	                     U = V_0 + V_1. \nonumber \\
\texttt{table\_max}\_a_2(U, X_1) &\leftarrow &   \texttt{value}(x_1, X_1), \\
					&& \texttt{table\_row}\_a_2(\_, X_1) \nonumber\\
					&&  U = \texttt{\#max}\{V : \texttt{table\_row}\_a_2(V, X_1) \}. \nonumber\\
\texttt{table}\_\texttt{info}(a_2, a_1, x_1,0,1) & \leftarrow&
\end{eqnarray*}
\end{example}
The relationship between the ASP-based computation and Algorithm~\ref{algo1} is established in the following lemma. 
\begin{lemma} \label{l11} 
Let us consider a DCOP $\cal M$,  an agent $a_i \in {\cal A}$, and a pseudo-tree $T$. 
Let $a_i$ be an agent with $C_i = \{a_{c_1}, \dots, a_{c_l}\}$ and  
$M_{a_{c_t}}$ be the encoded $UTIL^{a_i}_{a_{c_t}}$ for $1 \le t \le l$. Furthermore, 
let us
assume that $a_p$ is the parent of $a_i$, 
  $sep_i = \{x_{s_1}, \dots x_{s_k}\}$, and 
$R_{a_i} = \{f_{r_1}, \dots, f_{r_k'}\}$ and $scp(f_{r_j}) = \{x_{r_{j_1}},\dots, x_{r_{j_w}}\}$ for $1 \leq j \leq k'$.
Then, the program $\Pi_{a_i} \cup I_{a_i} \cup (\bigcup_{a \in C_i} M_a)$ 
has a unique answer set $A$ and 
\begin{itemize} 
\item  $\texttt{table\_row}\_a_i(u, v_{s_1}, \dots, v_{s_k}) \in A$ iff 
there exists a value combination $X$ for variables in $scp(JOIN_{a_i}^{P_i})$ such that  $JOIN_{a_i}^{P_i}(X) = u$ where 
$\{ x_{s_1} = v_{s_1}, \dots, x_{s_k} =  v_{s_k}\} \subseteq X$ and $u \ne -\infty$; and  
%
%
\item  $\texttt{table\_max}\_a_i(u, v_{s_1}, \dots, v_{s_k}) \in A$ iff
%
$UTIL_{a_i}^{P_i}(x_{s_1} = v_{s_1}, \dots, x_{s_k} =  v_{s_k}) = u$ and $u \not = -\infty$.
\end{itemize} 
\end{lemma} 
\begin{proof}
Since $I_{a_i}$ is safe and 
$\Pi_{a_i} \cup I_{a_i} \cup (\bigcup_{a \in C_i} M_a)$ is a positive program, 
it has a unique answer set.  


By the definition of answer set, $\texttt{table\_row}\_a_i(u, v_{s_1}, \dots, v_{s_k}) \in A$ iff that there exists a rule $r$ of the form ~\eqref{summingrule} such that $\texttt{table\_row}\_a_i(u, v_{s_1}, \dots, v_{s_k}) $ is the head of $r$. 
It means that there exists a value assignment $Y$ for the variables occurring in $r$ such that 
the following conditions hold:
\begin{itemize} 
\item $\{ x_{s_1} = v_{s_1}, \dots, x_{s_k} =  v_{s_k}\} \subseteq Y$;

\item for each $1 \le j \le k'$, there exists $v_{r_j}  \ne \#inf$ such that 
$f_{r_j}(v_{r_{j_1}},\dots, v_{r_{j_w}}) =  v_{r_j}$  
and $\{x_{r_{j_1}}= v_{r_{j_1}},\dots, x_{r_{j_w}} = v_{r_{j_w}}\} \subseteq Y$; 
and 
\item for each $1 \le t \le l$, there exists $u_{c_t}$ such that 
 $\mathtt{table\_max\_a_{c_t}} (u_{c_t}, v_{c_{t_1}}, \dots, v_{c_{t_z}}) \in A$ and 
 $\{x_{c_{t_1}}= v_{c_{t_1}},\dots, x_{c_{t_z}} = v_{c_{t_z}}\} \subseteq Y$. 
By the construction of the algorithm,  
 $\mathtt{table\_max\_a_{c_t}} (u_{c_t}, v_{c_{t_1}}, \dots, v_{c_{t_z}}) \in A$
 implies that 
 $UTIL_{a_{c_t}}^{a_i}(x_{c_{t_1}} = v_{c_{t_1}}, \dots, x_{c_{t_z}} =  v_{c_{t_z}}) = u_{c_t}$ and $u_{c_t} \ne -\infty$.
\end{itemize} 
The conclusion of the first item follows directly from the definitions of the UTIL  message and  the $\oplus$ operator  (Definitions~\ref{def3}-\ref{joinoperation}) and the above conditions. 



%

The second item of the lemma follows from the first item, the condition ${V_{r_{1}} \,!\!= \,\#\mathit{inf}, \cdots, V_{r_{k'}} \,!\!= \,\#\mathit{inf}}$ in the rule~\eqref{summingrule}, and  Definition~\ref{projectionoperation}.  
%
\end{proof}

Lemma~\ref{l11} implies that Phase 2 of ASP-DPOP computes the same UTIL messages as DPOP, except that UTIL messages in ASP-DPOP omit the value assignments whose associated utilities are $-\infty$.

\subsubsection{Computing the VALUE Message: Phase 3} 
\label{subsubVALUE}
 
Each ASP-DPOP agent computes the optimal value for its variables and sends an encoded VALUE message to its children. 
The process is described by the following Prolog rule:
 
 
 
{\small
\begin{verbatim}
% Phase 3: VALUE Propagation Phase
perform_Phase_3(ReceivedVALUEMessage):-
         separatorlist(Separator),
         constraintlist(ConstraintList),
         generate_VALUE_ASP(Separator,ConstraintList),
         solve_answer_set2(ReceivedVALUEMessage, Answer),
         send_message_to_children(a_i, value, Answer).
\end{verbatim}
}
\noindent In particular, the agent $a_i$: 
\begin{itemize}
\item Waits to receive the encoded VALUE message, denoted by $M'_{P_i}$, from its parent agent $P_i$. If the agent $a_i$ does not have a parent, i.e., it is the root of the tree, we set $Received\mathit{VALUE}Message = [ \: ]$;
\item Retrieves $sep_i$ (i.e., $Separator$) computed in Phase 2;
\item Retrieves $R_{a_i}$ 
(i.e., $ConstraintList$) computed in Phase 2;
\item Executes the rule $\texttt{generate\_{VALUE}\_ASP}/2$ to create an ASP program, denoted by $I'_{a_i}$, 
from the separator list, the constraint list, and the information from $M_{a_k}$ where $a_k \in C_i$ from Phase 2.
 Assume that  
\begin{itemize}
\item $sep_i = \{x_{s_1}, \dots, x_{s_{k}}\}$ (i.e., $Separator = [x_{s_1}, \dots, x_{s_k}]$ is the separator list of $a_i$);
\item $R_{a_i} = \{f_{r_1}, \dots, f_{r_{k'}}\}$ and $scp(f_{r_j}) = \{x_{r_{j_1}},\dots, x_{r_{j_w}}\}$ for $1 \leq j \leq k'$ (i.e., $ConstraintList = [f_{r_1}, \dots, f_{r_{k'}}]$);
\item $C_{i} = \{a_{c_1}, \dots, a_{c_l}\}$ and each $UTIL_{a_{c_t}}^{a_i}$ has $x_{c_{t_1}}, \dots, x_{c_{t_{z}}}$ as its dimensions for $1 \leq t \leq l$; and 
\item The set of variables owned by the agent $a_i$ is $\alpha_i = \{x_{i_1}, \dots, x_{i_q}\}$.
\end{itemize} 
$\texttt{generate\_{VALUE}\_ASP}/2$ creates the logic program  $I'_{a_i}$ with following rules:
%

{\small \begin{equation}
\begin{array}{ll}
\{\texttt{row}(U, X_{i_1}, \dots, X_{i_q})\} & \leftarrow \texttt{solution}(\alpha(x_{s_1}), x_{s_1}, X_{s_1}),  \label{valueoptimize} \\
&\cdots \\
&\texttt{solution}(\alpha(x_{s_k}), x_{s_k}, X_{s_k}), \\
&\texttt{table\_max}\_a_i(U, X_{s_1}, \dots X_{s_k}),  \\
&f_{r_1}(V_{r_1}, X_{r_{1_1}},\dots, X_{r_{1_w}}),   \\
&\cdots \\
&f_{r_{k'}}(V_{r_{k'}}, X_{r_{k'_1}},\dots, X_{r_{k'_w}}),  \\
&\texttt{table\_max}\_a_{c_1}(U_{c_1}, X_{c_{1_1}},\dots, X_{c_{1_{z}}}),  \\
&\cdots \\
&\texttt{table\_max}\_a_{c_l}(U_{c_l}, X_{c_{l_1}},\dots, X_{c_{l_{z}}}),  \\
&U == V_{r_1} + \dots + V_{r_{k'}} + U_{c_1} + \dots + U_{c_l}. \\
\end{array}
\end{equation}}

{\small \begin{equation}
\begin{array}{ll}
\hspace{1cm} \leftarrow \naf 1\{\texttt{row}(U, X_{i_1}, \dots, X_{i_q})\} 1. \label{valueoptimizeconstr} \\
\end{array}
\end{equation}}

{\small
	\begin{equation}
	\begin{array}{ll}
	\texttt{solution}(a_i, x_{i_1}, X_{i_1}) &\leftarrow \texttt{row}(U, X_{i_1}, \dots, X_{i_q}). \label{definesolution}\\
	& \cdots \\
	\texttt{solution}(a_i, x_{i_q}, X_{i_q}) &\leftarrow \texttt{row}(U, X_{i_1}, \dots, X_{i_q}). \
	\end{array}
	\end{equation}
	}

%
Intuitively, the rule of the form~\eqref{valueoptimize} and the constraint of the form~\eqref{valueoptimizeconstr} select an optimal row based on:
{\bf (i)} The computation as done in Phase 2 (i.e., using the facts of the form $\texttt{table\_max}\_a_i$ that are stored in Phase 2), and 
{\bf (ii)} (for non-root agent only) the VALUE message that is received from its parent (i.e., facts of the form 
$\texttt{solution}/3$). The selected optimal row will define the optimal value of the  variables using  rules of the form  \eqref{definesolution}. Similar argument for the safety of $I_{a_i}$   allows us to conclude that $I_{a_i}'$ is also a {\em safe} program.


\item Executes $\texttt{solve\_answer\_set2}/2$, that executes the ASP solver to compute an answer set of the program $\Pi_{a_i} \cup I'_{a_i} \cup M'_{P_i} \cup (\bigcup_{a \in C_{i}} M_a)$. From that answer set, it collects all facts of the form $solution(a,x,v)$ and returns them as $Answer$---i.e., the encoding of the VALUE message from the agent $a_i$ to its child agents;
\item Executes $\texttt{send\_message\_to\_children}/3$ where it sends $value$ message with $Answer$ as its data to each  child agent (i.e., executing the respected clauses $\texttt{send\_message}/4$ multiple times).
\end{itemize}


\begin{example}\label{ex/I'a2}
As an example, we refer to the DCOP in Example~\ref{exdcop} to illustrate $I'_{a_2}$ generated by the agent $a_2$.  $Received\mathit{UTIL}Messages$ for the agent $a_2$ is the set of ASP facts given in Example~\ref{utilmessexample}, 
$Separator = [x_1]$, 
$ConstraintList = [x1\_cons\_x2]$, 
%
 and $\alpha_2 = \{x_2\}$. The program $I'_{a_2}$ includes following rules:
\begin{eqnarray*}
\{\texttt{row}(U, X_2)\} &\leftarrow& \texttt{solution}(a_1, x_1, X_1), \nonumber \\
			&		&\texttt{table\_max}\_a_2(U, X_1), \nonumber \\
			&		& x_1\_cons\_x_2(U_0, X_1, X_2), \nonumber \\
			&                & \texttt{table\_max}\_a_3(U_1, X_1, X_2) \nonumber \\
			&		  & U == U_0 + U_1. \\
			 &\leftarrow& \naf 1\{\texttt{row}(U, X_2)\} 1. \\
\texttt{solution}(a_2, x_2, X_2) & \leftarrow & \texttt{row}(U, X_2).			 
\end{eqnarray*}
\end{example}

\begin{lemma} \label{l12} 
Let us consider a DCOP $\cal M$, and an agent $a_i \in {\cal A}$ in a pseudo-tree $T$. 
Let $a_i$ be an agent with $C_i = \{a_{c_1}, \dots, a_{c_l}\}$ and  
$M_{a_{c_t}}$ be the encoding of $UTIL^{a_i}_{a_{c_t}}$ for $1 \le t \le l$. Furthermore, 
assume that $P_i$ is the parent agent of the agent $a_i$, 
  $sep_i = \{x_{s_1}, \dots x_{s_k}\}$,   
$R_{a_i} = \{f_{r_1}, \dots, f_{r_k'}\}$ where $scp(f_{r_j}) = \{x_{r_{j_1}},\dots, x_{r_{j_w}}\}$ for $1 \leq j \leq k'$, $\alpha_i = \{x_{i_1}, \dots, x_{i_q}\}$,  
and $M_{P_i}'$ encodes $\mathit{VALUE}^{a_i}_{P_i}$. Let  $Q = \Pi_{a_i} \cup I'_{a_i} \cup M_{P_i}' \cup (\bigcup_{a \in C_i} M_a))$. Then,  
\begin{itemize} 
\item For each answer set $A$ of $Q$, 
the assignment $x_{i_j} = v_{i_j}$ where $\texttt{solution}(a_i, x_{i_j}, v_{i_j}) \in A$ for $1\le j \le q$ belongs to 
a solution of $\cal M$; and 
\item if $x_{i_1} = v_{i_1}, \ldots, x_{i_q} = v_{i_q}$ is a value assignment for variables in $\alpha_i$ that belongs to 
a solution of $\cal M$, which contains $VALUE^{a_i}_{P_i}$,  then  
$Q$ has an answer set $A$ containing $\{\texttt{solution}(a_i, x_{i_j}, v_{i_j}) \mid 1 \le j \le q\}
\cup M_i'$. 
\end{itemize} 
\end{lemma} 
\begin{proof}
Based on the construction of $I'_{a_i}$, it is possible to see that there exists at least  one rule of the form~\eqref{valueoptimize} in $Q$.
Observe that if the agent $a_i$ is  the root of the pseudo-tree $T$. Then, $M_{P_i}' = \emptyset$ and 
the rule~\eqref{valueoptimize} does not contain the atom of the form $\texttt{solution}(a, x, v)$. 
Since the program is safe and positive, we have that $Q$ is consistent. 

Because of the rule~\eqref{valueoptimizeconstr}, each answer set  $A$ of $Q$ contains exactly one atom of the form 
$\texttt{row}(u, v_{i_1}, \ldots, v_{i_q})$.  Also, from the rule~\eqref{valueoptimize}, we have that if 
$\texttt{row}(u, v_{i_1}, \ldots, v_{i_q}) \in A$ then there exists some 
$\texttt{table\_max}\_a_i(u, v_{s_1},\ldots,v_{s_k}) \in A$ which indicates that 
$u$ is the optimal utility corresponding to the assignment $x_{s_i} = v_{s_i}$ for $1 \le i \le k$ (Lemma~\ref{l11}). 
From the correctness of DPOP, this means that $\texttt{row}(u, v_{i_1}, \ldots, v_{i_q})$ 
encodes an optimal value assignment for variables owned by the agent $a_i$. This proves the first item. 

Assume that  $x_{i_1} = v_{i_1}, \ldots, x_{i_q} = v_{i_q}$ is a value assignment for variables in $\alpha_i$ that belongs to 
a solution of $\cal M$, which contains $VALUE^{a_i}_{P_i}$. Then, by the completeness of DPOP and Lemma~\ref{l11}, this implies that there exists some  $\texttt{table\_max}\_a_i(u, v_{s_1},\ldots,v_{s_k})$ such that $VALUE^{a_i}_{P_i}$ contains 
$x_{s_i} = v_{s_i}$ for $1 \le i \le k$. As such, there must exist the values for $f_{r_j}(.)$ and 
$\texttt{table\_max}\_a_{c_t}(.)$ 
such that there is a rule of the form~\eqref{valueoptimize} whose head is $\texttt{row}(u, v_{i_1}, \ldots, v_{i_q})$.
This means that $Q$ has an answer set containing $\texttt{row}(u, v_{i_1}, \ldots, v_{i_q})$, which proves 
the second item of the lemma.  
\end{proof} 

\subsubsection{ASP-DPOP}

In this section, we will show the clause for ASP-DPOP agents to perform Phase 1, Phase 2, and Phase 3 consecutively.  
For simplicity, we omit the fragment of code of ASP-DPOP agents that allow them to cooperate with each other to select one agent as the root of the pseudo-tree---since it depends on the scores that are assigned to agents, according to a heuristic function. 

Let us remind that, if an agent $a_k$ is the root of the pseudo-tree, a fact of the form $\texttt{parent}(master, a\_k)$ will be added to $\Pi_{a_k}$. After an agent is selected as the root of the pseudo-tree, each agent will execute the clause 
$\texttt{start}$. 
Considering an agent $a_i$, the clause $\texttt{start}$ of the agent $a_i$ is described as follows:

{\small 
\begin{verbatim}
% Perform Phase 1, Phase 2, and Phase 3
start:-
     (parent(master, a_i) ->
          generate_tree(master, [])
          ;
          wait_message(Parent, a_i, tree, Data),
          generate_tree(Parent, Data)
     ),
     (children(_) ->
          get_UTILMessages_from_all_children(ReceivedUTILMessages),
          perform_Phase_2(ReceivedUTILMessages)
          ;
          perform_Phase_2([])
     ),
     (parent(master, a_i) ->
          perform_Phase_3([])
          ;
          wait_message(Parent, a_i, value, ReceivedVALUEMessage),
          perform_Phase_3(ReceivedVALUEMessage)   
     ).
\end{verbatim}
}

\noindent
In particular, each agent $a_i$:
\begin{itemize}
\item Checks whether  $a_i$ is the root of the pseudo-tree; this is realized by checking whether the fact of the form 	
	$\texttt{parent}(master, a\_i)$ is in $\Pi_{a_i}$: 
	\begin{itemize}
	\item  If the agent $a_i$ is the root of the pseudo-tree, it will execute $\texttt{generate\_tree}(master, [\,])$ that is defined in Section~\ref{subsubphase1}; otherwise, 
	\item If the agent $a_i$ is not the root of the pseudo-tree, it will 
	execute $\texttt{wait\_message}(Parent, a\_i, tree, Data)$. Upon receiving the $tree$ message from the agent $Parent$ who is later assigned as its parent agent, the agent $a_i$ will execute 
	$\texttt{generate\_tree}(Parent, Data)$ that is defined in Section~\ref{subsubphase1}. 
\end{itemize}   
\item Checks whether the agent $a_i$ is (resp.~is not) a leaf of the pseudo-tree (i.e., this is realized by checking whether the fact of the form $\texttt{children}/1$ is not (resp.~is) in $\Pi_{a_i}$, respectively): 
\begin{itemize}
\item If the agent $a_i$ is not a leaf of the pseudo-tree, it executes 	
	{\small $\texttt{get\_UTILMessages\_from\_all\_children(ReceivedUTILMessages)}$}. 

Intuitively, the clause $\texttt{get\_UTILMessages\_from\_all\_children/1}$ iteratively executes
 $\texttt{wait\_message}(From, a\_i, util, Data)$ until the agent $a_i$ receives all $util$ messages from all of its child agents. The contents (i.e., $Data$) of all $util$ messages are combined into $\texttt{ReceivedUTILMessages}$. 
Then the agent $a_i$ executes $\texttt{perform\_Phase\_2(ReceivedUTILMessages)}$ that is defined in Section~\ref{phase2ASPDPOP}; otherwise, 
\item If the agent $a_i$ is a leaf of the pseudo-tree, it executes $\texttt{perform\_Phase\_2([\,])}$ that is defined in Section~\ref{phase2ASPDPOP}. 
\end{itemize}
\item Checks whether the agent $a_i$ is the root of the pseudo-tree or not: 
\begin{itemize}
\item  If the agent $a_i$ is the root of the pseudo-tree, it executes $\texttt{perform\_Phase\_3([\,])}$ that is defined in Section~\ref{subsubVALUE}; otherwise, 
\item  If the agent $a_i$ is not the root of the pseudo-tree, it executes 
	$\texttt{wait\_message}(Parent, a\_i, value, \mathit{ReceivedVALUEMessage})$ to wait for $value$ message from its parent agent. Then the agent $a_i$ executes $\texttt{perform\_Phase\_3(ReceivedVALUEMessage)}$ that is defined in Section~\ref{subsubVALUE}. 
\end{itemize}
\end{itemize}

\section{Theoretical Analysis}\label{sec:analysis}
In this section, we present some theoretical properties of ASP-DPOP including its soundness, completeness, and complexity.

\subsection{Soundness and Completeness}

The soundness and completeness of ASP-DPOP follow from Lemmas~\eqref{l11}--\eqref{l12} 
and the soundness and completeness of DPOP.

\begin{proposition}
ASP-DPOP is sound and complete in solving DCOPs.
\end{proposition}
\begin{proof}
Let us summarize how ASP-DPOP solves a DCOP $\cal M$:
\begin{list}{$\bullet$}{\leftmargin=12pt}
\item In Phase 1, each ASP-DPOP agent runs distributed DFS to generate a pseudo-tree. At the end of this phase, the information about the parent, pseudo-parents, and child agents of each agent $a_i$ are added to $\Pi_{a_i}$ via facts of the forms 
$\texttt{parent}/1$, $\texttt{pseudo\_parent}/1$ and $\texttt{children}/1$, respectively;   

\item In Phase 2, each ASP-DPOP agent $a_i$: {\bf (i)} waits to receive the encoding of  UTIL messages from all of its child agents (for non-leaf agents only), and  {\bf (ii)} generates the ASP program $I_{a_i}$ to compute its encoded UTIL message as an answer set of $\Pi_{a_i} \cup I_{a_i} \cup Received\mathit{UTIL}Messages$;

\item In Phase 3, each ASP-DPOP agent $a_i$:
{\bf (i)} waits to receive the encoded  VALUE message from its parent agent (for non-root agent only), and
	{\bf (ii)} generates the ASP program $I'_{a_i}$ to compute its encoded VALUE message as an answer set of $\Pi_{a_i} \cup I'_{a_i} \cup M'_{P_i} \cup Received\mathit{UTIL}Messages$;
\end{list}
The soundness and completeness of ASP-DPOP follows from the soundness and completeness of DPOP and the following observations: 
\begin{list}{$\bullet$}{\leftmargin=12pt}
\item Phase 1 of ASP-DPOP generates a possible pseudo-tree of $\cal M$. 
\item Assuming that ASP-DPOP and DPOP use the same pseudo-tree $T$ then: 
\begin{list}{$-$}{\leftmargin=12pt}
\item Phase 2 of ASP-DPOP computes the same UTIL messages as DPOP except that they omit the value assignments whose associated utilities are $-\infty$ (Lemma~\ref{l11}). However, for DPOP, ignoring those value assignments in UTIL message will not affect DCOP solution since such value assignments are not included in any solution (i.e., otherwise the total utility is $-\infty$).      
\item Phase 3 of ASP-DPOP computes all possible solutions (VALUE messages) as DPOP (Lemma~\ref{l12}). 
\end{list}
\end{list}
\end{proof}

\subsection{Complexity}
Given $d = \max_{1\leq i \leq n}|D_i|$ and $w^* = \max_{1\leq i \leq n}|sep_i|$ where $n$ is the number of agents,\footnote{$w^*$ is also known as the \emph{induced width} of a pseudo-tree~\cite{DBLP:books/0016622}.} we have the following properties:
\begin{property}
The number of messages required by ASP-DPOP is bounded by $\mathcal{O}(n)$. 
\end{property}
\begin{proof}
In ASP-DPOP, one can observe that: {\bf (1)} 
Phase 1 requires a linear number of messages in $n$ (Lemma~\ref{lemma/phase1linear});
%
%
{\bf (2)} Phase 2 requires $(n-1)$ UTIL messages; and
{\bf (3)} Phase 3 requires $(n-1)$ VALUE messages.
Thus, the number of messages required by ASP-DPOP is bounded by $\mathcal{O}(n)$.
\end{proof}

\begin{property}\label{prop2}
The size of messages required by ASP-DPOP is bounded by $\mathcal{O}(d^{w^*})$.
\end{property}

\begin{proof}
In ASP-DPOP, 
\begin{list}{$\bullet$}{\topsep=1pt \parsep=0pt \itemsep=1pt \leftmargin=12pt}
\item Phase 1 produces messages whose size is linear in $n$. This is because 
the {\em tree} message is of the form~$send\_message(a\_i, Next\_Agent, tree, Data)$ where the content $(Data)$ that dominates the size of the {\em tree} message is the set of visited agents whose size is linear in $n$;
\item Phase 2 produces encoded UTIL messages; each message consists of:
	{\bf (i)} a fact of the form $\texttt{table\_max}\_a_i$ for each cell in the corresponding UTIL message in DPOP
	where its associated optimal utility is not $-\infty$, and 
	{\bf (ii)} $|sep_i|$ facts of the form $\texttt{table}\_\texttt{info}$. Therefore, the size of encoded UTIL messages is bounded by $\mathcal{O}(d^{w^*})$ as the bounded size of UTIL messages in DPOP~\cite{petcu:05}; and 
\item Phase 3 produces encoded VALUE messages; each message consists of a fact of the form $\texttt{solution}/3$ for each value assignment of a variable in the corresponding VALUE message in DPOP. Therefore, Phase 3 produces encoded VALUE messages whose sizes are bounded by
$\mathcal{O}(|\mathcal{X}|) = \mathcal{O}(n)$ as we assume each agent owns exactly one variable.
\end{list}
Thus, the size of messages required by ASP-DPOP is bounded by $\mathcal{O}(d^{w^*})$
\end{proof}
\begin{property}
The memory requirements in ASP-DPOP are exponential and bounded by $\mathcal{O}(d^{w^*})$.
\end{property}
\begin{proof}
In ASP-DPOP:
\begin{list}{$\bullet$}{\topsep=1pt \parsep=0pt \itemsep=1pt \leftmargin=12pt}
\item In Phase 1, the memory requirements are bounded by $\mathcal{O}(n)$ 
because it needs to keep track of the set of visited agents and the set of its neighboring agents; 
\item In Phase 2, the memory requirements are bounded by $\mathcal{O}(d^{w^*})$ since, in computing the answer set of $P = \Pi_{a_i} \cup I_{a_i} \cup Received\mathit{UTIL}Messages$, the ASP solver needs to ground all the rules of the forms~\eqref{summingrule} and \eqref{maxrule}, and these dominate the number of other facts or ground instances of other rules in $P$;
\item In Phase 3, the memory requirements are bounded by $\mathcal{O}(d^{w^*})$ since, in computing the answer set of $P' = \Pi_{a_i} \cup I'_{a_i} \cup M'_{P_i} \cup Received\mathit{UTIL}Messages$, the ASP solver needs to ground all rules of the form~\eqref{valueoptimize}, where the number of facts of the form $\texttt{table\_max}\_a_i$ is bounded by $\mathcal{O}(d^{w^*})$ (see Property~\ref{prop2}). Moreover, the number of such ground instances dominates the number of other facts and ground instances of other rules in $P'$.
\end{list}
Thus, the memory requirement in ASP-DPOP is exponential and bounded by $\mathcal{O}(d^{w^*})$.
\end{proof}

\section{Experimental Results}
\label{sec:expr}
The goal of this section is to provide an experimental evaluation of ASP-DPOP. In particular, we
compare ASP-DPOP against the original DPOP as well as other 
three  implementations of complete DCOP solvers: \emph{Asynchronous Forward-Bounding (AFB)}, \emph{Hard Constraint-DPOP (H-DPOP)}, and \emph{Open-DPOP (ODPOP)}.
AFB~\cite{gershman:09}  is a complete search-based algorithm to solve DCOPs.
H-DPOP~\cite{kumar:08} is a complete DCOP solver that, in addition,  propagates hard constraints to prune the search space.
\emph{ODPOP}~\cite{ODPOP:PetcuF06} is an optimization algorithm for DCOPs, which combines some advantages of both search-based algorithms and dynamic-programming-based algorithms.
For completeness of the paper, in this section, we will first provide some background about these
three solvers, discuss about FRODO platform~\cite{leaute:09}---a publicly-available implementation of DPOP, AFB, and ODPOP---and then provide some experimental results.

\subsection{Background on AFB}\label{sec:AFB}
The \emph{asynchronous forward-bounding algorithm (AFB)}~\cite{gershman:09}, to the best of our knowledge, is the most recent complete search-based algorithm to solve DCOPs. 
AFB makes use of  a \emph{Branch and Bound} scheme to identify a complete value assignment that minimizes the aggregate utility of all constraints.
To do so, agents expand a partial value assignment as long as the lower bound on its aggregate utility does not exceed the global bound, which is the aggregate utility of the best complete value assignment found so far. 

In AFB, the state of the search process is represented by a data structure called \emph{Current Partial Assignment (CPA).} The CPA starts empty at some initializing agent, which records the value assignment of its own variable and sends the CPA to the next agent. The aggregate utility of a CPA is the accumulated utility of constraints involving the value assignment it contains. Each agent, upon receiving a CPA, adds a value assignment of its own variable such that the CPA's aggregate utility will not exceed the global upper bound. If it cannot find such an assignment, it backtracks by sending the CPA to the last assigning agent, requesting that agent to revise its value assignment.  

Agents in AFB process and communicate CPAs asynchronously. An agent that succeeds to expand the value assignment of the received CPA sends forward copies of the updated CPA, requesting all unassigned agents to compute lower bound estimates of the aggregate utility of the current CPA. The assigning agent will receive these estimates asynchronously over time, and use them to update the lower bound of the CPA. 
Using these bounds, the assigning agent can detect if any expansion of this partial value assignment in the current CPA will cause it to exceed the global upper bound, and in such cases it will backtrack. Additionally, a time stamp mechanism for forward checking is used by agents to determine the most updated CPA and to discard obsolete CPAs.

\subsection{Background on H-DPOP}\label{sec:H-DPOP}
In H-DPOP~\cite{kumar:08}, the authors consider how to leverage the hard constraints that may exist in the problem in a dynamic programming framework, so that only feasible partial assignments are computed, transmitted, and stored~\cite{kumar:08}. To this end, they encode combinations of assignments using \emph{Constrained Decision Diagrams (CDDs).} Basically, CDDs eliminate all inconsistent assignments and only include utilities corresponding to value combinations that are consistent. The resulting algorithm, H-DPOP, a hybrid algorithm that is based on DPOP, uses CDDs to rule out infeasible assignments, and thus compactly represents UTIL messages. 

\smallskip
A CDD $\mathcal{G} = \langle \Gamma, G\rangle$ encodes the consistent assignments for a set of constraints $\Gamma$ in a rooted direct acyclic graph $G = (V, E)$ by means of constraint propagation. A node in $G$ is called a CDD node. The terminal nodes are either \emph{true} or \emph{false} implying consistent or inconsistent assignment, respectively. By default, a CDD represents consistent assignments omitting the \emph{false} terminal. 

The H-DPOP algorithm leverages the pruning power of hard constraints by using CDDs to effectively reduce the message size. As in DPOP, H-DPOP has three phases: the pseudo-tree construction, the bottom-up UTIL propagation, and top-down VALUE propagation. The pseudo-tree construction and VALUE propagation phases are identical to ones of DPOP. In the UTIL propagation phase, the UTIL message, instead of being a multidimensional matrix, is a \emph{CDDMessage}. 
\begin{definition}\label{CDDMessage}
A CDDMessage $M_i^j$ sent by an agent $a_i$ to agent $a_j$ is a tuple $\langle \vec{u}, \mathcal{G}\rangle$ where $\vec{u}$ is a vector of utilities, and $\mathcal{G}$ is a CDD defined over variables in $sep_i$. The set of constraints for $\mathcal{G}$ is $\Gamma = \{f_j \mid scp(f_j) \subseteq sep_i\}$.
\end{definition} 
In the UTIL propagation phase, H-DPOP defines different JOIN and PROJECTION operations. 
Observe that,  based on Definition~\ref{CDDMessage}, an H-DPOP agent $a_i$ can  access a constraint whose scope is
a  subset of its separator, but that is not owned by $a_i$ itself. For example, considering the DCOP in Example~\ref{exdcop}, in H-DPOP, the UTIL message sent from the agent $a_3$ to the agent $a_2$ will have information about the constraint $x1\_cons\_x_2$ that is not owned by the agent $a_3$ since $scp(x1\_cons\_x_2) = \{x_1, x_2\} \subseteq sep_3$. This might be undesirable
in situations where distribution of the computation is tied to privacy of information.

\subsection{Background on ODPOP}

\emph{ODPOP}~\cite{ODPOP:PetcuF06} is an optimization algorithm for DCOP, which combines some advantages of both search-based algorithms and dynamic-programming-based algorithms. ODPOP always uses linear size messages, which is similar to search, and typically generates as few messages as DPOP. It does not always incur the worst complexity which is the same with the complexity of DPOP, and on average it saves a significant amount of computation and information exchange. This is achieved because agents in ODPOP use a \emph{best-first order} for value exploration and an \emph{optimality criterion} that allows them to prove optimality without exploring all value assignments for variables in their separator.  
ODPOP also has 3 phases as DPOP:

\smallskip
\noindent {\bf Phase 1} ({\em DFS traversal}) is the same with Phase 1 in DPOP to construct a pseudo-tree.

\smallskip
\noindent  {\bf Phase 2} (\emph{ASK/GOOD}) is an iterative, bottom-up utility propagation process where each agent repeatedly sends ASK messages to its child agents, asking for valuations (\emph{GOODs}), until it is able to compute the \emph{suggested optimal value assignment} (GOOD) for variables in its separator. It then sends that GOOD, together with the respective utility that is obtained in the subtree rooted at this agent, as a GOOD message to its parent agent. This phase finishes when the root received enough GOODs to determine the optimal value assignment for its variables. 

In more detail, in Phase 2, any child agent delivers to its parent agent a sequence of GOOD messages, each of which explores a different value assignment for variables in its separator, together with the corresponding utility. In addition, ODPOP uses a method to propagate GOODs so that every agent always reports its GOODs in oder of \emph{non-increasing utility}, provided that all of their child agents also follow this order. We can see that, DPOP agents receives all GOODs that are grouped in single messages (i.e., UTIL messages). In contrast, ODPOP agents send GOODs on demand (i.e., when it receives an ASK message) individually and asynchronously as long as GOODs have non-increasing utilities. 

As a consequence, each ODPOP agent $a_i$ can determine when it has received \emph{enough} GOODs from its child agents in order to be able to determine a GOOD to send to its parent agent $P_i$. At that time, $a_i$ will not send ASK message any more since 
any additional received GOODs  will not affect the GOOD that was determined. If $a_i$ later receives more ASK message from $P_i$ for having next GOOD, $a_i$ will continue to request more GOODs from its child agents until it can determine the next GOOD to report to $P_i$.  

Since GOODs are always reported in order of non-increasing utility, the first GOOD that is generated at the root agent is the optimal value assignment for its variable. As a consequence, the root agent will be able to generate this solution without having to consider all value assignments for its variables.

\smallskip
\noindent  {\bf Phase 3} ({\em VALUE propagation}) is similar to Phase 3 in DPOP. 
The root agent initiates the top-down VALUE propagation by sending a VALUE message to its child agents, informing them about its optimal value assignment for its variables. Subsequently, each agent $a_{i'}$, upon receiving a VALUE message, will determine its optimal value assignment for its variables based on the computation (in Phase 2) of the first GOOD message generated whose associated value assignment is consistent with the one in the received VALUE message.

\subsection{Discussion on FRODO Platform}

In our experiment, we will compare the  performance of ASP-DPOP against DPOP,  AFB, and ODPOP; in particular, we use
the implementation of the latter three systems that is  publicly available in the  \emph{FRODO} platform~\cite{leaute:09}.
It is important to observe that, at the implementation level, all DCOP solvers that are implemented within FRODO follow the \emph{simplifying assumption} that each agent owns exactly one variable. This assumption is common practice in the DCOP literature~\cite{modi:05,petcu:05,gershman:09,ottens:12}. 
However, agents in DCOP problems used in our experiments own multiple variables. 
We will discuss in this subsection the pre-processing technique (i.e., \emph{decomposition}, also known as \emph{virtual agents}) that FRODO uses to transform a general DCOP with multiple variables per agent into a new DCOP with one variable per agent.

FRODO creates a virtual agent for each variable in a DCOP. A distinct variable is assigned to each virtual agent, so that this formulation satisfies the simplifying assumption. 
We say that a virtual agent $a'_i$ belongs to a real agent $a_i$ in DCOP if the virtual agent $a'_i$ owns a variable that is owned by the real agent $a_i$. 
In FRODO, the solving algorithm is executed on each virtual agent, while intra-agent messages 
(i.e., messages are passed between virtual agents that belong to the same real agent)
are only simulated and discounted in the calculation of computation cost (e.g., number of messages and messages' size).


 Let $M$ be a DCOP with multiple variables per agent, and $M'$ be a new DCOP with one variable per agent that is constructed from $M$.  Let us assume that we apply DPOP to solve both $M$ and $M'$, using  the same heuristics to construct the pseudo-trees. We can observe that each node in the pseudo-tree used to solve $M'$ represents  a virtual agent, while each node in a pseudo-tree of ASP-DPOP represents  a real agent. It is possible to see that, the number of inter-agent messages (i.e., messages that are passed between virtual agents that belong to different real agents) produced in solving $M'$ may be greater than the number of UTIL messages produced in solving $M$, 
 depending on their respective pseudo-trees. Therefore, to minimize the total number of inter-agent messages, FRODO constructs pseudo-trees where virtual agents that belong to the same real agent stay as close as possible to each other.     

It is important to summarize that, to handle a general DCOP with multiple variables per agent, FRODO first transforms it into a new DCOP with one variable per agent (introducing virtual agents), and then executes the resolution algorithms on each agent of the new DCOP.
To the best of our knowledge, there is not any formal discussion about the relationship between pseudo-trees whose nodes represent real agents and pseudo-trees whose nodes represent virtual agents.
However, it is believed that given a pseudo-tree $T'$ whose nodes represent virtual agents, there always exists a pseudo-tree $T$ whose nodes represent real agents such that $T$ is compatible with $T'$. 
Intuitively, by compatible we mean that it is possible to construct $T$ from $T'$ as follows:
\begin{itemize}
\item If the root of $T'$ is a node representing the virtual agent $a'_i$ that belongs to a real agent $a_i$, the root of $T$ is the node representing $a_i$; and
\item If there is at least one tree edge (resp.~back edge) connecting two nodes that represent virtual agents $a'_{i_1}$ and $a'_{i_2}$ in $T'$, there is a tree edge (resp.~back edge) connecting the two nodes that represent real agents $a_{i_1}$ and $a_{i_2}$ in $T$ such that $a'_{i_1}$ and $a'_{i_2}$ belong to $a_{i_1}$ and $a_{i_2}$ respectively.    
\end{itemize}
It is worth to notice that, in our experiments, we ensure that all algorithms use the same heuristics to construct their pseudo-trees. We also observe that all pseudo-trees that are constructed using ASP-DPOP are compatible with the corresponding pseudo-trees that are constructed using FRODO.
 
\subsection{Experimental Results}  
We implement two versions of the ASP-DPOP algorithm---one that uses ground programs, which we call 
\emph{``ASP-DPOP (facts),''} and one that uses non-ground programs, which we call 
\emph{``ASP-DPOP (rules).''} In addition, as the observation made about H-DPOP in Section~\ref{sec:H-DPOP}, we also implemented a variant of H-DPOP, called PH-DPOP, which stands for 
\emph{Privacy-based H-DPOP,} that restricts the amount of information that each agent can access to the amount common in most DCOP algorithms, including DPOP and ASP-DPOP. Specifically, agents in PH-DPOP can only access their own constraints and, unlike H-DPOP, cannot access their neighboring agents' constraints. 

In our experiments, we compare both versions of ASP-DPOP against DPOP~\cite{petcu:05}, H-DPOP~\cite{kumar:08}, PH-DPOP,  AFB~\cite{gershman:09}, and ODPOP~\cite{ODPOP:PetcuF06}.
We use a publicly-available implementation of DPOP, AFB, and ODPOP~\cite{leaute:09} and an implementation of H-DPOP provided by the authors. We ensure that all algorithms use the same heuristics to construct their pseudo-trees for fair comparison. We measure the runtime of the algorithms using the simulated runtime metric~\cite{sultanik:07}. All experiments are performed on a Quadcore 3.4GHz machine with 16GB of memory. If an algorithm fails to solve a problem, it is due to memory limitations; other
types of failures are specifically  stated.  We conduct our experiments on random graphs~\cite{erdos59a}, where we systematically modify the domain-independent parameters, and on comprehensive optimization problems in power networks~\cite{gupta:13}.

\begin{table*}[htbp]
\small
\begin{minipage}[h]{1\textwidth}
\resizebox{0.90\linewidth}{!}{
\begin{tabular}{|c|r|r|r|r|r|r|r|r|r|r|r|r|}
\cline{1-13}
\multirow{2}{*}{$|\mathcal{X}|$} & \multicolumn{2}{c|}{DPOP} & \multicolumn{2}{c|}{H-DPOP} & \multicolumn{2}{c|}{PH-DPOP} & \multicolumn{2}{c|}{AFB} & \multicolumn{2}{c|}{ODPOP} & \multicolumn{2}{c|}{ASP-DPOP}  \\
  & Solved & Time & Solved  & Time & Solved  & Time & Solved  & Time & Solved  & Time & Solved  & Time \\
\cline{1-13}
\cline{1-13}
5 & 100\% & 36 & 100\% & 28 & 100\% & 31 & 100\% & 20 & 100\% & 31& 100\% & 779 \\
10 & 100\% & 204 & 100\% & 73 & 100\% & 381 & 100\%& 35 & 100\%& 164& 100\% & 1,080 \\
15 & 86\% & 39,701 & 100\% & 148 & 98\% & 67,161 & 100\%& 53 & 100\%& 3,927&100\% & 1,450 \\
20 & 0\% & - & 100\% & 188 & 0\% & - & 100\%& 73 &74\%\footnote{ODPOP cannot solve 13 instances (out of 50 instances) in this experiment in which there are 12 instances unsolved due to timeout and 1 instance unsolved due to memory limitation.}& 242,807& 100\% & 1,777 \\
25 & 0\% & - & 100\% & 295 & 0\% & - & 100\% & 119 & 0\% & - & 100\% & 1,608 \\
\cline{1-13}
\multicolumn{5}{c}{} \\
\cline{1-13}
\multirow{2}{*}{$p_1$} & \multicolumn{2}{c|}{DPOP} & \multicolumn{2}{c|}{H-DPOP} & \multicolumn{2}{c|}{PH-DPOP}& \multicolumn{2}{c|}{AFB} & \multicolumn{2}{c|}{ODPOP} & \multicolumn{2}{c|}{ASP-DPOP} \\
 & Solved & Time & Solved  & Time & Solved  & Time & Solved  & Time & Solved  & Time & Solved  & Time \\
\cline{1-13}
\cline{1-13}
0.4 & 100\% & 1,856 & 100\% & 119 & 100\% & 2,117 & 100\% & 46 & 100\%& 1,819& 100\% & 1,984 \\
0.5 & 100\% & 13,519 & 100\% & 120 & 100\% & 19,321 & 100\% & 50 & 100\%& 2,680 & 100\% & 1,409 \\
0.6 & 94\% & 42,010 & 100\% & 144 & 100\% & 54,214 & 100\% & 51 & 100\% & 3,378 & 100\% & 1,308\\
0.7 & 56\% & 66,311 & 100\% & 165 & 88\% & 131,535 & 100\% & 54 & 100\% & 8,063 & 100\% & 1,096 \\
0.8 & 20\% & 137,025 & 100\% & 164 & 62\% & 247,335 & 100\% & 60 & 100\%& 36,748&  100\% & 1,073 \\
\cline{1-13}
\end{tabular}
}

\vspace{1em}
\resizebox{0.90\linewidth}{!}{
\begin{tabular}{|c|r|r|r|r|r|r|r|r|r|r|r|r|}
\cline{1-13}
\multirow{2}{*}{$|D_i|$} & \multicolumn{2}{c|}{DPOP} & \multicolumn{2}{c|}{H-DPOP} & \multicolumn{2}{c|}{PH-DPOP} & \multicolumn{2}{c|}{AFB} & \multicolumn{2}{c|}{ODPOP} & \multicolumn{2}{c|}{ASP-DPOP} \\
 & Solved & Time & Solved & Time & Solved & Time & Solved  & Time & Solved  & Time & Solved  & Time  \\
\cline{1-13}
\cline{1-13}
4 & 100\% & 782 & 100\% & 87 & 100\% & 1,512 & 100\% & 46 & 100\% & 285 & 100\% & 1,037 \\
6 & 90\% & 28,363 & 100\%& 142 & 98\% & 42,275 & 100\% & 50 &100\%& 4,173& 100\% & 1,283 \\
8 & 14\% & 37,357 & 100\%& 194 & 52\% & 262,590 & 100\% & 60 & 98\% & 71,512 & 100\% & 8,769 \\
10 & 0\% & - & 100\% & 320 & 8\% & 354,340 & 100\% & 70 & 78\%\footnote{ODPOP cannot solve 11 instances (out of 50 instances) in this experiment in which there are 10 instances unsolved due to timeout and 1 instance unsolved due to memory limitation.}& 227,641& 100\% & 29,598 \\
12 & 0\% & - & 100\% & 486 &  0\% & - & 100\% & 82 & 30\%\footnote{ODPOP cannot solve 35 instances (out of 50 instances) in this experiment in which there are 29 instances unsolved due to timeout and 6 instance unsolved due to memory limitation.}& 343,756& 100\% & 60,190 \\
\cline{1-13}
\multicolumn{5}{c}{} \\
\cline{1-13}
\multirow{2}{*}{$p_2$}  & \multicolumn{2}{c|}{DPOP} & \multicolumn{2}{c|}{H-DPOP} & \multicolumn{2}{c|}{PH-DPOP} & \multicolumn{2}{c|}{AFB} & \multicolumn{2}{c|}{ODPOP} & \multicolumn{2}{c|}{ASP-DPOP} \\
 & Solved & Time & Solved & Time & Solved  & Time & Solved  & Time & Solved  & Time & Solved  & Time\\
\cline{1-13}
\cline{1-13}
0.3 & 90\% & 38,114 & 100\% & 464 & 76\%  & 189,431 & 100\% & 103 & 84\%\footnote{ODPOP cannot solve 8 instances (out of 50 instances) in this experiment due to timeout.}& 221,515 & 18\% & 120,114 \\
0.4 & 86\% & 48,632 & 100\% & 265 & 84\% & 107,986 & 100\% & 71 & 94\%\footnote{ODPOP cannot solve 3 instances (out of 50 instances) in this experiment due to timeout.}& 109,961& 86\% & 50,268 \\
0.5 & 94\% & 38,043 & 100\% & 161  & 96\% & 71,181 & 100\% & 57 & 100\%& 14,790& 92\% & 4,722 \\
0.6 & 90\% & 31,513 & 100\% & 144  & 98\% & 68,307 & 100\% & 52 & 100\%& 13,519 & 100\% & 1,410 \\
0.7 & 90\% & 39,352 & 100\% & 128  & 100\% & 49,377 & 100\% & 48 & 100\%& 1,730& 100\% & 1,059 \\
0.8 & 92\% & 40,526 & 100\% & 112  & 100\% & 62,651 & 100\% & 57 & 100\% & 1,137 & 100\% & 1,026 \\
\cline{1-13}
\end{tabular}
}
\end{minipage}
\caption{Experimental Results on Random Graphs}
\label{randomGraphResults}
\end{table*}

\smallskip \noindent \textbf{Random Graphs:} We create an $n$-node network, whose constraint density $p_1$ produces $\lfloor n\cdot(n-1)\cdot p_1\rfloor$ edges in total~\cite{erdos59a}. 
In our experiments, we vary the number of variables $|\mathcal{X}|$, the domain size $|D_i|$, the constraint density $p_1$, and the constraint tightness $p_2$. For each  experiment, we vary only one parameter and fix the rest to their ``default'' values: $|\mathcal{A}| = 5, |\mathcal{X}| = 15, |D_i| = 6, p_1 = 0.6, p_2 = 0.6$. The timeout is set to $5\times10^6$ ms. Table~\ref{randomGraphResults}  shows the percentage of instances solved (out of 50 instances) and the average simulated runtime (in ms) for the solved instances. We do not show the results for ASP-DPOP (rules), as the utilities in the utility table are randomly generated, leading to no differences w.r.t. ASP-DPOP (facts). We make the following observations:

\begin{itemize}
\item ASP-DPOP is able to solve more problems than DPOP and is faster than DPOP when the problem becomes more complex (i.e.,~increasing $|\mathcal{X}|$, $|D_i|$, $p_1$, or $p_2$). The reason is that ASP-DPOP is able to prune a significant portion of the search space thanks to hard constraints. ASP-DPOP does not need to explicitly represent the rows in the UTIL table that are infeasible, unlike DPOP. The size of the search space pruned increases as the complexity of the instance grows, resulting in a larger difference between the runtime of DPOP and ASP-DPOP. 

\item H-DPOP is able to solve more problems and solve them faster than DPOP, PH-DPOP, and ASP-DPOP. The reason is that agents in H-DPOP utilize more information about the neighbors' constraints to prune values. In contrast, agents in ASP-DPOP and PH-DPOP only utilize information about their own constraints to prune values and agents in DPOP do not prune any values. 

\item ASP-DPOP is able to solve more problems and solve them faster than PH-DPOP. The reason is that agents in PH-DPOP, like agents in H-DPOP, use constraint decision diagram (CDD) to represent their utility tables, and it is expensive to maintain and perform join and project operations on this data structure. In contrast, agents in ASP-DPOP are able to capitalize on highly efficient ASP solvers to maintain and perform operations on efficient data structures thanks to their highly optimized grounding techniques and use of portfolios of heuristics.

\item AFB is able to solve more problems and solve them faster than every other algorithm. We attribute this observation mainly to the relatively small number of variables in this experiment---i.e., the maximum number of variables in this experiment is 25 (see the first table in Table~\ref{randomGraphResults}).   

\item ASP-DPOP is able to solve more problems and solve them faster than ODPOP when the problem becomes more complex (i.e., increasing $|\mathcal{X}|, |\mathcal{D}_i|, p_1,$ or $p_2$). The reason is that ODPOP does not prune the search space based on hard constraints, unlike  ASP-DPOP. On one hand, ODPOP intuitively sends each row of UTIL messages per time on demand and uses optimality criteria to prove optimality without exploring all value assignments for the respective variables. However, these techniques are not as efficient as pruning the search space in ASP-DPOP when the problem becomes more complex. Thus, ODPOP reaches a timeout in most of its unsolvable problems. It is also worth to observe that there are some problems that  ODPOP cannot solve due to memory limitations. We attribute this to the fact that ODPOP  maintains in its search space \emph{infeasible} value assignments that result in a utility equal to $- \infty$, and thus the search space of ODPOP is not as optimized as that of ASP-DPOP.     

\end{itemize}
{\center
\begin{table*}[htbp]
\small
\begin{minipage}[h]{1\textwidth}
\resizebox{0.90\linewidth}{!}{
\begin{tabular}{|c|r|r|r|r|}
\cline{1-5}
\multirow{2}{*}{$|\mathcal{X}|$} & \multicolumn{2}{c|}{AFB} & \multicolumn{2}{c|}{ASP-DPOP}  \\
  & Solved  & Time & Solved  & Time \\
\cline{1-5}
\cline{1-5}
150& 100\% & 31,156  & 100\% & 37,862 \\
200& 100\% & 117,913  & 100\% & 115,966 \\
 250& 0\%\footnote{AFB cannot solve any instance (out of 50 instances) in this experiment due to timeout.} & -  & 100\% & 298,361 \\
\cline{1-5}
\end{tabular}
}
\end{minipage}
\caption{Additional Experimental Results on Random Graphs}
\label{randomGraphResults_extend}
\end{table*}
}

\smallskip \noindent \textbf{Additional Experiment Results on Random Graphs:} We claimed earlier that AFB is able to solve more problems and solve them faster than every other algorithm, mainly due to the relative small number of variables in the experiments reported in Table~\ref{randomGraphResults}. To directly confirm such  claim, we extended our experiments on random graphs, by increasing the number of variables (i.e., $|\mathcal{X}| \geq 150$) and keeping the other parameters to their ``default'' values (i.e., $|\mathcal{A}| = 5, |D_i| = 6, p_1 = 0.6, p_2 = 0.6$).\footnote{We thank one of the reviewers for his/her suggestion to have this additional experiment on random graphs.} The timeout is also set to $5\times10^6$ ms. Table~\ref{randomGraphResults_extend} shows the percentage of instances solved (out of 50 instances) and the average simulated runtime (in ms) for the solved instances. The runtime results for DPOP, H-DPOP, PH-DPOP, and ODPOP are not included in Table~\ref{randomGraphResults_extend} because they run out of memory in solving all of the problems in this domain.\footnote{It is worth to note that H-DPOP runs out of memory while constructing its CDDs in solving all of the problems in this domain.}
We observe that ASP-DPOP is able to solve more problems than AFB (i.e., when $|\mathcal{X}| = 250$) and solve them faster than AFB when $|\mathcal{X}| \geq 200$. We attribute this observation mainly to the large number of variables in this experiment. We also notice that AFB can scale up to solve problems of up to $|\mathcal{X}| = 200$ (such scalability will not be seen in the experiment on power network problems described below). The main reason is that the problems in our experiment on random graphs are ``purely hard'' with the default values $p_1 = 0.6$ and $p_2 = 0.6$. This means that the size of the set of complete feasible value assignments, which are complete value assignments that do not result in a utility of $+ \infty$, is small (about less than $5$ in all of the problems in this domain). AFB  backtracks much earlier before it can achieves a complete feasible value assignment. As a result, AFB can solve problems with number of variables up to 200 before it exceeds the time out threshold.       

\begin{figure*}[htbp]
	\centering
	\parbox{2.25in}{
		\centering \footnotesize
		\includegraphics[width=2.25in]{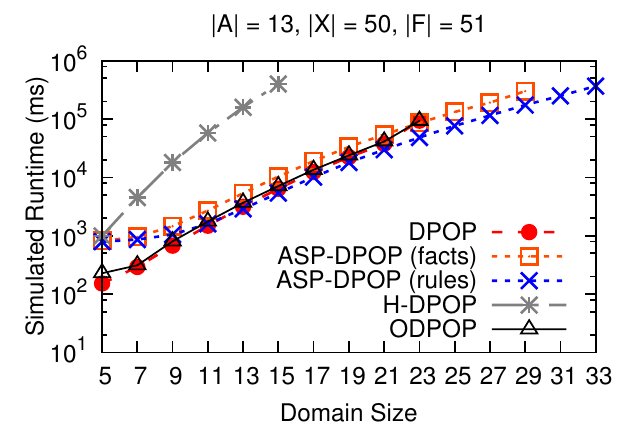}
		\\ \vspace{-0.5em} \hspace{1em} (a) 13 Bus Topology} \hspace{-0em}
	\parbox{2.25in}{
		\centering \footnotesize
		\includegraphics[width=2.25in]{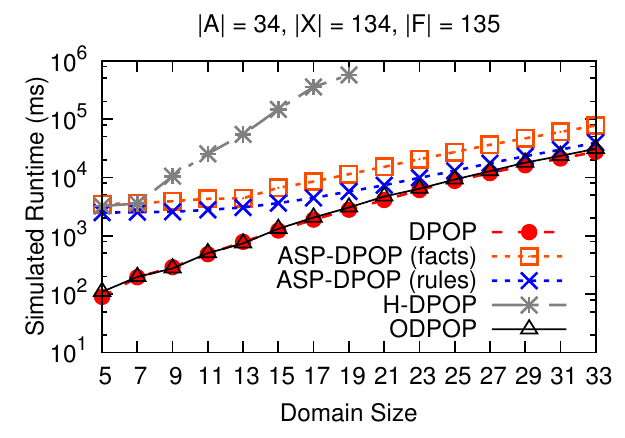}
		\\ \vspace{-0.5em} \hspace{1em} (b) 34 Bus Topology} \hspace{-0em}
	\parbox{2.25in}{
		\centering \footnotesize
		\includegraphics[width=2.25in]{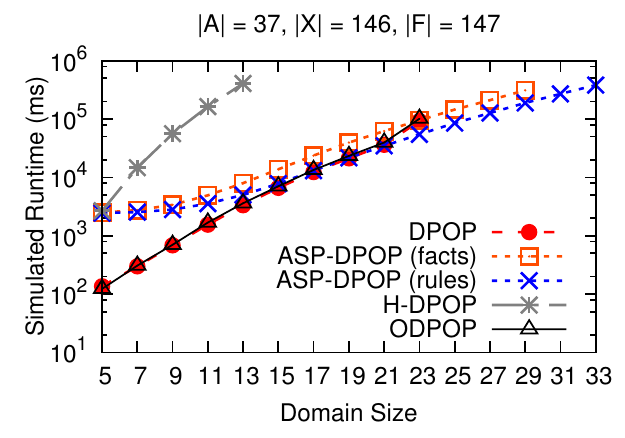}
		\\ \vspace{-0.5em} \hspace{1em} (c) 37 Bus Topology}
	\caption{Runtime Results on Power Network Problems}
	\label{fig:result-graphs}
\end{figure*}

\begin{table*}[htbp]
	\centering \small
	\begin{tabular}{|c|r|r|r|r|r|r|r|r|}
		\cline{1-9}
		\multirow{2}{*}{$|D_i|$} & \multicolumn{4}{c|}{13 Bus Topology}  & \multicolumn{4}{c|}{34 Bus Topology}  \\
		& 5 & 7 & 9 & 11 & 5 & 7 & 9 & 11  \\
		\cline{1-9}
		\cline{1-9}
		H-DPOP & 6,742 & 30,604 & 97,284 & 248,270 & 1,437 & 4,685 & 11,617 & 24,303 \\
		DPOP & 3,125 & 16,807 & 59,049 & 161,051 &  625 & 2,401 & 6,561 & 14,641  \\
		ODPOP & 6 & 6 & 6 & 6 &5 & 5 & 5 & 5 \\
		ASP-DPOP & 10 & 14 & 18 & 22 & 10 & 14 & 18 & 22   \\
		\cline{1-9}
		\cline{1-5}
		\multirow{2}{*}{$|D_i|$}  & \multicolumn{4}{c|}{37 Bus Topology} \\
		&  5 & 7 & 9 & 11  \\
		\cline{1-5}
		\cline{1-5}
		H-DPOP & 6,742 & 30,604 & 97,284 & 248,270 \\
		DPOP &  3,125 & 16,807 & 59,049 & 161,051 \\
		ODPOP & 6 & 6 & 6 & 6 \\
		ASP-DPOP & 10 & 14 & 18 & 22  \\
		\cline{1-5}
		
		\noalign{\smallskip}
		\multicolumn{9}{c}{(a) Largest UTIL Message Size} \\
		\noalign{\smallskip}
		\cline{1-9}
		\multirow{2}{*}{$|D_i|$} & \multicolumn{4}{c|}{13 Bus Topology}  & \multicolumn{4}{c|}{34 Bus Topology} \\
		& 5 & 7 & 9 & 11 &  5 & 7 & 9 & 11  \\
		\cline{1-9}
		\cline{1-9}
		H-DPOP & 19,936 & 79,322 & 236,186 & 579,790 & 20,810 & 57,554 & 130,050 & 256,330 \\
		DPOP & 9,325 & 43,687 & 143,433 & 375,859 &  9,185 & 29,575 & 73,341 & 153,923  \\
		ODPOP & 391 & 1,430 & 6,281 & 11,979 &2,197 & 4,122 & 12,124 & 12,870 \\
		ASP-DPOP & 120 & 168 & 216 & 264 &  330 & 462 & 594 & 726  \\
		\cline{1-9}
		\cline{1-5}
		\multirow{2}{*}{$|D_i|$}  & \multicolumn{4}{c|}{37 Bus Topology} \\
		&  5 & 7 & 9 & 11  \\
		\cline{1-5}
		\cline{1-5}
		H-DPOP & 38,689 & 133,847 & 363,413 & 836,167 \\
		DPOP &  17,665 & 71,953 & 215,793 & 531,025 \\
		ODPOP & 1,896 & 5,572 & 18,981 & 28,285 \\
		ASP-DPOP & 360 & 504 & 648 & 792 \\
		\cline{1-5}
		
		\noalign{\smallskip}
		\multicolumn{9}{c}{(b) Total UTIL Message Size} \\
	\end{tabular}
	\caption{Message Size Results on Power Network Problems} 
	\label{mesSizeSM}
\end{table*}

\smallskip \noindent \textbf{Power Network Problems:} A {\em customer-driven microgrid} (CDMG), one possible instantiation of the smart grid problem, has recently been shown to subsume several classical power system sub-problems (e.g.,~load shedding, demand response, restoration)~\cite{jain:12}. In this domain, each agent represents a node with consumption, generation, and transmission preferences, and a global cost function. Constraints include the power balance and no power loss  principles, the generation and consumption limits,  and the capacity of the power line between nodes. The objective is to minimize a global cost function. CDMG optimization problems are well-suited to be modeled with DCOPs due to their distributed nature. Moreover, as some of the constraints in CDMGs (e.g.,~the power balance principle) can be described in functional form, they can be exploited by ASP-DPOP (rules). For this reason, both ``ASP-DPOP (facts)'' and ``ASP-DPOP (rules)'' are used in this domain. 

We use three network topologies defined using the IEEE standards~\cite{IEEEStandards} and vary the domain size of the generation, load, and transmission variables of each agent from 5 to 31. The timeout is set to $10^6$ ms. Figure~\ref{fig:result-graphs} summarizes the runtime results. As the utilities are generated following predefined rules~\cite{gupta:13}, we also show the results for ASP-DPOP (rules). Furthermore, we omit results for PH-DPOP because they have identical runtime---the amount of information used to prune the search space is identical for both algorithms in this domain. We also measure the size of UTIL messages, where we use the number of values in the message as units, and the intra-agent UTIL messages (i.e., messages are passed between virtual agents that belong to the same real agent) are accounted  for fair comparison.
%
%
Table~\ref{mesSizeSM} tabulates the results. We did not measure the size of VALUE messages since they are significantly smaller than UTIL messages. It is also worth to report that the number of UTIL messages that FRODO produces (discounting all intra-agent UTIL messages) is equal to the number of UTIL messages that ASP-DPOP produced in all power network problems in our experiments.

The results in Figure~\ref{fig:result-graphs} are consistent with those shown earlier (except for AFB)---ASP-DPOP is slower than DPOP and ODPOP when the domain size is small, but it is able to solve more problems than DPOP and ODPOP. 
We observe that, in  Figure~\ref{fig:result-graphs}(b), DPOP is consistently faster than ASP-DPOP and is able to solve the same number of problems as ASP-DPOP. It is because the highest constraint arity in 34 bus topology is 5 while it is 6 in 13 and 37 bus topologies.
Unlike in random graphs, H-DPOP is slower than the other algorithms in these problems. The reason is that the constraint arity in these problems is larger and the expensive operations on CDDs grow exponentially with the arity. We also observe that ASP-DPOP (rules) is faster than ASP-DPOP (facts). The reason is that the former is able to exploit the interdependencies between constraints to prune the search space. Additionally, ASP-DPOP (rules) can solve more problems than ASP-DPOP (facts). The reason is that the former requires less memory since it prunes a larger search space and, thus, ground fewer facts. 

The runtime results for AFB are not included in Figure~\ref{fig:result-graphs}, since AFB exceeds the timeout in solving all of the problems in this domain; this contrasts to the results shown earlier for random graphs. The main reason is that the number of variables in the power network problems is large (i.e., $|\mathcal{X}|$ are $50, 134$, and $146$ in $13, 34,$ and $37$ bus topologies, respectively in Figure~\ref{fig:result-graphs}).

Finally, both versions of ASP-DPOP require smaller messages than both H-DPOP and DPOP. The reason for the former is that the CDD data structure of H-DPOP is significantly more complex than that of ASP-DPOP. The reason for the latter is that ASP-DPOP prunes portions of the search space while DPOP did not. In addition, since ASP-DPOP does not transform DCOP problems with multiple variables per agent to corresponding ones with one variable per agent, ASP-DPOP is able to exploit significantly more the interdependencies between constraints to prune the search space. 
Moreover, we can see that the largest UTIL message sizes in ODPOP are smaller than those of ASP-DPOP, but the total UTIL message sizes in ODPOP are larger than those of ASP-DPOP. The reason is that ODPOP sends only linear size message, but it needs to send many messages on demand.

\subsection{Discussions on ASP-DPOP} 

ASP-DPOP has been shown to be competitive with other algorithms in solving DCOPs in our experimental results. The benefits of using ASP-DPOP are accomplished by having ASP as its foundation. We will illustrate here the two main advantages of making use of ASP within ASP-DPOP: 
\begin{list}{}{\topsep=1pt \parsep=0pt \itemsep=1pt}
\item[1.] The use of the highly expressive ASP language to encode constraints in DCOPs; and 
\item[2.] The ability to harness the highly optimized ASP grounder and solver to prune the search space based on hard constraints. \end{list}
In the rest of this section, we further discuss these advantages and relate them to the observations 
drawn from the experiments. These considerations are followed by a discussion of how ASP-DPOP alleviates the simplifying assumption of having a single variable per agent. 
Finally, at the end of this section, we  analyze the privacy loss of ASP-DPOP.

The first advantage of using ASP within ASP-DPOP comes from the ability to use a very expressive logic
language to encode the constraints in a DCOP.
ASP-DPOP can represent constraint utilities as an implicit function instead of explicitly enumerating them.
Thus, ASP-DPOP is particularly suitable to encode  DCOPs whose constraint utilities are large and evaluated via implicit functions of the variables in their scopes (e.g., power network problems, smart grid problems). This can be seen clearly via Example~\ref{ex/expressiveness}.

\begin{example}
\label{ex/expressiveness}
Let us consider a constraint $f$ representing the power loss principle in a power network problem, where $scp(f) = \{x_{1\rightarrow2}, x_{2\rightarrow1}\}$ in which the domains of the variables $x_{1\rightarrow2}$ and $x_{2\rightarrow1}$ are $D_{1 \rightarrow 2 } = [0,2]$ and $D_{ 2 \rightarrow 1} = [-2,0]$, respectively.
Intuitively, the variable $x_{i \rightarrow j}$, where $i,j \in \{1,2\}, i \not = j$, indicates the amount of power that node $i$ transfers to (receives from) node $j$ if $x_{i \rightarrow j} \geq 0$ (resp.~$x_{i \rightarrow j} < 0$). For example, $x_{1 \rightarrow 2} = 1$ means that the node $1$ transfers $1$ unit of power to the node $2$, and $x_{2 \rightarrow 1} = -1$ specifies that the node $2$ receives $1$ unit of power from the node 1. By the power loss principle, if there is no loss, the amount of power transferred from one node is equal to the amount of power received in the other node (i.e., $x_{i \rightarrow j} + x _{j \rightarrow i} = 0$). However, if there is loss (i.e., $x_{i \rightarrow j} + x _{j \rightarrow i} \not = 0$), we assume that the cost (utility) of the power transmission is evaluated to be two times greater than the power unit loss. Formally, the utility of the constraint $f$ is given implicitly as a function: 

\begin{equation}
f(x_{1\rightarrow2}, x_{2\rightarrow1}) = 2 \times |x_{1\rightarrow2} + x_{2\rightarrow1}|.  
\end{equation}
\end{example}

\begin{figure}[htbp]
  \begin{minipage}[h]{0.3\textwidth}
    \centering      
    \begin{tabular}{|c|c|c|}
      \cline{1-3}
      $x_{1\rightarrow2} $ & $x_{2\rightarrow 1}$ & Utilities \\
      \cline{1-3}
      2 & -2 & 0\\
      2 & -1 & 2\\
      2 & 0 & 4\\
      1 & -2 & 2\\
      1 & -1 & 0\\
      1 &  0 & 2\\
      0 & -2 & 4\\
      0 & -1 & 2\\
      0 &  0 & 0\\
      \cline{1-3}
    \end{tabular}
    \\ \vskip0.2cm (a) Explicit Representation as a Utility Table
  \end{minipage}
  \hspace{0.3cm}
    \begin{minipage}[h]{0.6\textwidth}
 \begin{eqnarray}
 value(x_{1 \rightarrow 2}, 0..2) \leftarrow &  & \label{x12}\\
 value(x_{2 \rightarrow 1}, -2..0) \leftarrow &  & \label{x21}\\
f(2 * |V_1 + V_2|, V_1, V_2) &\leftarrow& value(x_{1 \rightarrow 2}, V_1), \nonumber \\
					   &                & value(x_{2 \rightarrow 1}, V_2). 		 
\end{eqnarray}
  \\ \vskip0.7cm \center(b) Implicit Representation as an Answer Set Program
  \end{minipage}
   \hspace{-0.25in}
  \caption{Different Encodings of Constraint $f$ in Example~\ref{ex/expressiveness}} \label{expressiveness}
\end{figure}

Figure~\ref{expressiveness}(a) enumerates all the utilities of the constraint $f$ explicitly in a utility table, and Figure~\ref{expressiveness}(b) presents an answer set program that models implicitly those utilities. We can see that while the utility table has 9 rows (i.e., the domain sizes of  $x_{1\rightarrow2}$ and $x_{2\rightarrow1}$ are 3), the answer set program consists of only 2 facts and 1 rule. If the domain sizes of $x_{1\rightarrow2}$ and $x_{2\rightarrow1}$ are $1000$ (e.g., $D_{1 \rightarrow 2} = [0,999]$ and $D_{1 \rightarrow 2} = [-999,0]$), the utility table would have ${\bf 1000^2}$ rows whereas the answer set program modeling implicitly such the same utilities still has {\bf 2 facts} and {\bf 1 rule} that are similar to ones in Figure~\ref{expressiveness}(b)---i.e., it only updates the 2 facts~\eqref{x12} and~\eqref{x21} as follows: 
 \begin{eqnarray}
 value(x_{1 \rightarrow 2}, 0\,..\,999) &\leftarrow  & \\
 value(x_{2 \rightarrow 1}, -999\,..\,0) & \leftarrow &  		 
\end{eqnarray}
As a consequence, using ASP within ASP-DPOP to encode DCOPs makes programs much more concise and compact. The encoding is declarative and can be easily extended and modified. Moreover, such encoding does not depend on the implementation of the  algorithms (e.g., DPOP or H-DPOP), making programs more flexible and understandable.  Specifically, if we change the algorithm to solve a DCOP, the Controller Module needs to be changed following the new algorithm, yet the Specification Module remains the same. 
In contrast, using imperative programming techniques, the ``ad-hoc" implementation that is employed within each local solver might require different encodings of DCOPs for different used algorithms and different propagators for different types of constraints. For example, H-DPOP implementation needs a different data structure from DPOP implementation to deal with hard constraints.  

The second advantage of using ASP as the foundation of ASP-DPOP is to harness the highly optimized ASP grounders
and solvers  to prune the search space, especially in the handling of  hard constraints. As an example, consider the power network problem whose objective is to minimize a global cost function.\footnote{The previous formalization of ASP-DPOP focuses on maximizing the cost function; the switch to minimization problems requires trivial changes to the design of ASP-DPOP.} 
A  DCOP that encodes such type of power network problems can be formulated in terms of cost-as-utility {\em minimization} rather than reward-as-utility {\em maximization}. Thus,  in this formulation   
the value assignments resulting in an infinite utility (i.e., $+\infty$) 
should not be included in any DCOP solution; such  value assignments are redundant and 
should be pruned. Example~\ref{ex/prune} shows how effectively an ASP grounder can prune the search space. 

\begin{example}
\label{ex/prune}
Consider a simple power network problem, where the aggregated cost needs to be minimized.
The problem has two nodes (nodes $1$ and $2$). Let us assume that agent $a_1$ and agent $a_2$, which are the node $1$ and the node $2$, own the variables $x_{1\rightarrow 2}$ and $x_{2 \rightarrow 1}$, respectively. 
These  are described in Example~\ref{ex/expressiveness}. The problem has one constraint $f$ representing the power loss principle,
analogously to what described  in Example~\ref{ex/expressiveness}. The only difference is that we do not
allow  losses in power transfers (i.e., if there is a loss, the corresponding cost is $+\infty$). Thus, the utility (cost) of the constraint $f$ now is evaluated as:
%
%
\begin{equation}
 f(x_{1\rightarrow2}, x_{2\rightarrow1})= 
  \begin{cases} 
   2 \times x_{1\rightarrow2} & \text{if }  x_{1\rightarrow2} + x_{2\rightarrow1} = 0 \\
   + \infty       & \text{otherwise } 
  \end{cases}
\end{equation} 
Figure~\ref{fig/prune} presents the ASP program\footnote{{\tt \#sup} is a special constant representing the largest possible value in the ASP language.} to compute the UTIL message sent from the agent 2 to the agent 1, assuming that the agent 1 is the root of the respective pseudo-tree (i.e., the separator set of the agent 2 is $sep_2 = \{x_{1\rightarrow 2}\}$). 
It is important to observe that, since the objective is to minimize a global cost function, the ASP in Figure~\ref{fig/prune} is produced differently from the one that is generated by $\texttt{generate\_{UTIL}\_ASP}/2$ described in~\ref{phase2ASPDPOP}. Specifically, the differences are: 
\begin{itemize}
\item The predicates of the form $\texttt{table\_min}\_a_i$ are used instead of ones of the form $\texttt{table\_max}\_a_i$; 
\item U = \texttt{\#min\{\dots\}} rather than U = \texttt{\#max\{\dots\}} in~\eqref{maxrule} (i.e., computing the {\tt minimal} utilities for each value combination of variables in the separator list); and 
\item The conditions ${V_{r_{1}} \,!\!= \,\texttt{\#sup}, \, \cdots \, , V_{r_{k'}} \,!\!= \,\texttt{\#sup}}$ are used instead of ${V_{r_{1}} \,!\!= \,\texttt{\#inf}, \cdots, V_{r_{k'}} \,!\!= \,\texttt{\#inf}}$ in~\eqref{summingrule} (i.e., atoms of the form $\texttt{table\_row}\_a_i(u, v_{s_1}, \dots, v_{s_k})$ where ${u = \texttt{\#sup}}$ (i.e., the respective utilities are $+\infty$) are not produced). 
\end{itemize}
As a consequence, the encoded UTIL messages consist of facts of the forms $\texttt{table\_min}\_a_i$ (instead of $\texttt{table\_max}\_a_i$) and $\texttt{table}\_\texttt{info}$.  

\begin{figure}[htbp]
    \begin{eqnarray}
 f(4, 2, -2) & \leftarrow  & \label{fig/prune1}\\
  f(2, 1, -1) & \leftarrow & \label{fig/prune2}\\
   f(0, 0, 0) & \leftarrow &\label{fig/prune3} \\
    f(\texttt{\#sup}, 0, -2) & \leftarrow &\label{fig/prune4} \\ 
      f(\texttt{\#sup}, 1, -2) & \leftarrow & \label{fig/prune5}\\
        f(\texttt{\#sup}, 0, -1) & \leftarrow & \label{fig/prune6}\\
          f(\texttt{\#sup}, 2, -1) & \leftarrow & \label{fig/prune7}\\
            f(\texttt{\#sup}, 1, 0) & \leftarrow & \label{fig/prune8}\\
              f(\texttt{\#sup}, 2, 0) & \leftarrow &\label{fig/prune9} \\
              value(x_{1 \rightarrow 2}, 0\,..\,2) &\leftarrow& \label{fig/prune10}\\
 value(x_{2 \rightarrow 1}, -2\,..\,0) &\leftarrow& \label{fig/prune11}\\
\texttt{table\_row}\_a_2(U, X_{1\rightarrow 2})  & \leftarrow & f(V_0,X_{1\rightarrow 2}, X_{2 \rightarrow 1}), \nonumber\\
 &	& V_0 \,\,!\!= \, \texttt{\#sup},\nonumber\\
 &	& U = V_0.\label{fig/prune12}\\
 \texttt{table\_min}\_a_2(U,X_{1\rightarrow 2}) & \leftarrow & value(x_{1 \rightarrow 2},X_{1\rightarrow 2}),\label{fig/prune13} \\
 && \texttt{table\_row}\_a_2(\_, X_{1\rightarrow2}),\nonumber \\
 &&  U = \texttt{\#min}\{V : \texttt{table\_row}\_a_2(V, X_{1\rightarrow 2}) \}. \nonumber
\end{eqnarray}
  \caption{ASP to Compute UTIL Message in Example~\ref{ex/prune}} \label{fig/prune}
\end{figure}
The 9 facts~\eqref{fig/prune1}-\eqref{fig/prune9} enumerate all utilities of the constraint $f$ in which the 6 facts~\eqref{fig/prune4}-\eqref{fig/prune9} are redundant since their corresponding utilities are $+\infty$. With DPOP, the total size of the search space for computing its UTIL message is 9, which corresponds to the 9 facts~\eqref{fig/prune1}-\eqref{fig/prune9}, since DPOP does not do pruning. However, with ASP-DPOP, the corresponding total size of the search space is 3 since \textsc{gringo}, due to the condition $V_0 \,!\!= \#sup$, grounds the rule~\eqref{fig/prune12} into only 3 facts:
     \begin{eqnarray}
 \texttt{table\_row}\_a_2(4, 2) & \leftarrow & \\
  \texttt{table\_row}\_a_2(2, 1) & \leftarrow & \\
   \texttt{table\_row}\_a_2(0, 0) & \leftarrow &
   \end{eqnarray}   
and an ASP solver will use these facts to generate the predicates $\texttt{table\_min}\_a_2$ based on the rule~\eqref{fig/prune13}. The different between the sizes of the search spaces of ASP-DPOP and DPOP are greater as the domain sizes of variables increase. For example, if the domain sizes of $x_{1\rightarrow2}$ and $x_{2\rightarrow1}$ are $1000$, the total search space of DPOP is ${\bf 1000^2}$ while the total search space of ASP-DPOP is just ${\bf 1000}$.    
\end{example}
As a consequence, and as clear from  our experiments, ASP-DPOP is able to prune a significant portion of the search space, thanks to hard constraints, whereas DPOP does not. Moreover, as seen in Example~\ref{ex/prune}, the size of the search space pruned increases as the complexity of the instance grows (i.e., increasing $|\mathcal{X}|$, $|D_i|$, $p_1$, or $p_2$). Thus, ASP-DPOP is able to solve more problems than DPOP and is faster than DPOP when the problem becomes more complex. 

The pruning power of the ASP grounders and solvers enables also the generation of smaller
UTIL messages in ASP-DPOP than those generated by DPOP.
Let us consider a UTIL message $M$ sent from an agent $a_i$ to an agent $a_j$. A value assignment of variables in $sep_i$ is \emph{admissible} if its corresponding optimal sum of utilities in the subtree rooted at $a_i$ is different than $-\infty$.\footnote{Or $+\infty$ for minimization problems.} 
In DPOP, $M$ consists of a utility, which is optimal, for each value assignment of variables in $sep_i$ (including both admissible and inadmissible value assignments). However, $M$ in ASP-DPOP consists of a utility, which is optimal and different from $- \infty$, for only each admissible value assignment of variables in $sep_i$. This is because such inadmissible value assignments will not be included in any DCOP solution (i.e., otherwise the global cost is $-\infty$).

We will not discuss in-depth technically what algorithms and computations are implemented within 
modern ASP grounders to optimize  the grounding process, since they are beyond  the scope of this paper. Readers who are interested in such algorithms and computations can find further information in~\cite{claspbook,KaufmannLPS16}. It is important to notice that such computations (e.g., for removing unnecessary rules and for omitting rules whose bodies cannot be satisfied) consume memory, take time, and are not trivial. Therefore, for DCOP problems with low constraint tightness, the runtime and memory that are used for those computations will dominate the runtime and memory that are saved from pruning the search space (e.g., see the row $p_2 = 0.3$ in Table~\ref{randomGraphResults}). This also explains why ASP-DPOP is slower than DPOP when the problem becomes less complex (i.e.,~decreasing $|\mathcal{X}|$, $|D_i|$, $p_1$, or $p_2$). Specifically, from the trend while decreasing $p_2$ in Table~\ref{randomGraphResults}, ASP-DPOP will not be able to compete with DPOP for cases where $p_2 \leq 0.3$.

\smallskip
The fact that ASP-DPOP solves DCOP problems with multiple variables per agent directly, without transforming them to 
problems with one variable per agent, deserves some discussions. It is easy to see that ASP-DPOP agents need to consider more variables and thus more constraints. As a result, there are more interdependencies between constraints for ASP-DPOP to exploit. If the constraint tightness is high, the size of the search space pruned increases significantly. This can be seen in our power network experiment. On the other hand, dealing with more variables and more constraints also increases the search space. Therefore, if the constraint tightness does not provide sufficient pruning, the portion of the search space pruned does not properly balance the increase in the size of  the search space; this may lead  ASP-DPOP to require more  memory than DPOP in solving such problems. This situation can be seen in experimental results on random graphs (i.e., decreasing $p_2$). Solving DCOPs with multiple variables per agent without transforming them to problems with a single variable per agent was also investigated in~\cite{DBLP:conf/aaai/Fioretto0P16}.

\smallskip

Maintaining privacy is a fundamental motivation for the use of DCOP. A detailed analysis of privacy loss in DCOP for some existing DCOP algorithms, including DPOP, can be found in~\cite{PrivacyLoss:GreenstadtPT06}. For ASP-DPOP, it is not difficult to realize that DPOP and ASP-DPOP have the same privacy loss. The reason is that the content of UTIL messages (resp.~VALUE messages) in DPOP---that are given under the tabular form (which are similar to those given under multi-dimensional matrix form)---is identical to the content of the UTIL messages (resp.~VALUE messages) in ASP-DPOP---that are given in facts form. In fact, anything that is inferred from the fact form (in UTIL and VALUE messages of ASP-DPOP) can be inferred from the tabular form (in the respective messages of DPOP), and vice versa anything is inferred from tabular form can be inferred from fact form as well.


\section{Related Work}\label{sec:relatedwork}
The use of declarative programs, specifically logic programs, for reasoning in multi-agent domains is not new. Starting with some seminal papers~\cite{kow}, various authors have explored the use of several different flavors of logic programming, such as normal logic programs and abductive logic programs, to address cooperation between agents~\cite{kakas,toni,gelfond,laima}. Some proposals have also explored the combination between constraint programming, logic programming, and formalization of multi-agent domains~\cite{newtplp,vlahavas02,our10,DovierFP10b}. Logic programming has been used in modeling multi-agent scenarios involving agents knowledge about other's knowledge~\cite{BaralGPS10a}, computing models in the logics of knowledge~\cite{BaralGPS10b}, multi-agent planning~\cite{SonPS09} and formalizing negotiation~\cite{SakamaSP11}. ASP-DPOP is similar to the last two applications in that (\emph{i})~it can be viewed as a collection of agent programs; (\emph{ii})~it computes solutions using an ASP solver; and (\emph{iii})~it uses message passing for agent communication. A key difference is that ASP-DPOP solves multi-agent problems formulated as constraint-based models, while the other applications solve problems formulated as decision-theoretic and game-theoretic models. 

Researchers have also developed a framework that integrates declarative techniques with off-the-shelf constraint solvers to partition large constraint optimization problems into smaller subproblems and solve them in parallel~\cite{liu:12}. In contrast, DCOPs are problems that are naturally distributed and cannot be arbitrarily partitioned. 

ASP-DPOP is able to exploit problem structure by propagating hard constraints and using them to prune the search space efficiently. This reduces the memory requirement of the algorithm and improves the scalability of the system. Existing DCOP algorithms that also propagate hard and soft constraints to prune the search space include H-DPOP that propagates exclusively hard constraints~\cite{kumar:08}, BrC-DPOP that propagates branch consistency~\cite{fioretto:14b}, and variants of BnB-ADOPT~\cite{yeoh:10,gutierrez:12,gutierrez:11} that maintains soft-arc consistency~\cite{bessiere:12,gutierrez:12b,gutierrez:13}. A key difference is that these algorithms require algorithm developers to explicitly implement the ability to reason about the hard and soft constraints and propagate them efficiently. In contrast, ASP-DPOP capitalizes on general purpose ASP solvers to do so.

\section{Conclusions and Future Work}
\label{sec:conclusion}
In this paper, we explored the benefits of using  logic programming techniques as a platform to provide complete solutions of  DCOPs. Our proposed logic programming-based algorithm, ASP-DPOP, is able to solve more problems and solve them faster than DPOP, its imperative programming counterpart. Aside from the ease of modeling, each agent in ASP-DPOP also capitalizes on highly efficient ASP solvers to automatically exploit problem structure (e.g.,~prune the search space using hard constraints). Experimental results show that ASP-DPOP is faster and can scale to larger problems than a version of H-DPOP (i.e., PH-DPOP) that 
maintains the level of privacy similar to that of ASP-DPOP. 
These results highlight the strengths of a declarative programming paradigm, where explicit model-specific pruning rules are not necessary. In conclusion, we believe that this work contributes to the DCOP community, where we show that the declarative programming paradigm is a promising new direction of research for DCOP researchers, as well as the ASP community, where we demonstrate the applicability of ASP to solve a wide array of multi-agent problems that can be modeled as DCOPs. 

\smallskip
In future work, we will explore two directions to deepen the use of logic programming in solving DCOPs:  

\begin{itemize}
\item \emph{Logic programming under different semantics}: We will consider the advantages of other logic programming paradigms  in solving DCOPs. One possibility  is to use Constraint Logic Programming (CLP)~\cite{JAFFAR1994503} instead of ASP. Since CLP is a merger of two declarative paradigms---constraint solving and logic programming---it seems well-suited to solve DCOPs. A preliminary investigation~\cite{LePSY14} has shown that this technique can dramatically decrease run time.

\item \emph{Different representation of messages}: We observe that the messages used in DPOP, and even ASP-DPOP, are represented explicitly---i.e., they are multi-dimensional matrices in DPOP and facts in ASP-DPOP. One of the reasons for this is that each agent performs the inference process for its subtree, enumerates explicitly all the results, and sends them to other agents. We are interested in investigating  algorithms where agents coordinate with others via messages that are logic programs (e.g., ASP or CLP clauses). Specifically, in such an algorithm, each agent does the inference partially, for some specific interesting value assignment, and without enumerating  all results. The rest of the computation will be encoded as logic programs and  passed to other agents. Some agent who performs the complete inference process will propagate the search space based on the rules in the received messages as logic programs. We believe this will reduce the search space and  the run time.  
\end{itemize}
  
\section*{Acknowledgment}
This research is partially supported by NSF  grants  HRD-1345232 and DGE-0947465. The views and conclusions contained in this document are those of the authors and should not be interpreted as representing the official policies, either expressed or implied, of the sponsoring organizations, agencies, or the U.S. government. We would  like to thank Akshat Kumar for sharing with us his implementation of H-DPOP.

\bibliography{../bib/dcopAcronym,../bib/bibfile,../bib/bib2010,../bib/smartgridAcronym}

\end{document}